%% file: main.tex
\documentclass{article}

\usepackage[preprint]{template}
\usepackage{amsthm,amsmath,amssymb}
\usepackage{graphicx}
\usepackage{natbib}
\usepackage[hidelinks,breaklinks]{hyperref}
\usepackage{subfig}
\usepackage[utf8]{inputenc} 
\usepackage[T1]{fontenc}    
\usepackage{url}            
\usepackage{booktabs}       
\usepackage{tabularx}   
\usepackage{amsfonts}       
\usepackage{nicefrac}       
\usepackage{microtype}      
\usepackage[table]{xcolor}   
\usepackage{float}
\usepackage{algorithmic}
\usepackage{algorithm}

\usepackage{tikz}
\usetikzlibrary{arrows.meta,calc,positioning,shapes.geometric}
\definecolor{gapred}{rgb}{0.80, 0.25, 0.15}
\definecolor{cOurs}{HTML}{c4553a}
\definecolor{cRPNI}{HTML}{3d5a80}
\definecolor{cAlergia}{HTML}{2a8f82}
\definecolor{cGray}{HTML}{8a8478}
\definecolor{cLight}{HTML}{d5cfc5}
\definecolor{cOursBg}{HTML}{f5e3df}
\definecolor{cRPNIBg}{HTML}{dfe6ed}
\definecolor{cAlergiaBg}{HTML}{ddeeed}
\definecolor{cGroupHd}{HTML}{f0eeeb}
\definecolor{cBest}{HTML}{fef3e0}

\graphicspath{{figs/}{outputs/figures/}}

\newtheorem{theorem}{Theorem}

\newtheorem{definition}[theorem]{Definition}

\newtheorem{proposition}[theorem]{Proposition}
\newtheorem{remark1}[theorem]{Remark}
\newenvironment{remark}{\begin{remark1} \rm}{\end{remark1}}

\hypersetup{
    colorlinks=true,
    linkcolor=green!50!blue,
    citecolor=green!50!blue,
    urlcolor=green!50!blue
}

\title{Divergence-Guided Particle Swarm Optimization}

\author{%
  Kleyton da Costa\thanks{Corresponding author. This work was done during the author's period as M.Sc student in the Pontifical Catholic University of Rio de Janeiro} \\
  Pontifical Catholic University of Rio de Janeiro\\
  Rio de Janeiro, Brazil \\
  \texttt{kleyton.vsc@gmail.com} \\
  \And
  Bernardo Modenesi\\
  University of Utah\\
  Salt Lake City, US \\
  \texttt{bernardo.modenesi@utah.edu} \\
  \And
  Ivan F.M. Menezes\\
  Pontifical Catholic University of Rio de Janeiro\\
  Rio de Janeiro, Brazil \\
  \texttt{ivan@puc-rio.br} \\
  \And
  Hélio Lopes\\
  Pontifical Catholic University of Rio de Janeiro\\
  Rio de Janeiro, Brazil \\
  \texttt{lopes@inf.puc-rio.br} \\
}

\begin{document}

\maketitle

\begin{abstract}
   Particle Swarm Optimization (PSO) is susceptible to premature convergence when the swarm collapses around the global best, particularly on multimodal landscapes in higher dimensions. We propose Divergence-guided PSO (DPSO), which augments the velocity update with a modulation term that repels particles whose personal bests have converged near the global best. The repulsion is gated by a Gaussian similarity kernel, which we prove is equivalent to an exponentially decaying function of the KL divergence between Gaussian-embedded personal and global bests, connecting the mechanism to the family of $f$-divergences and providing a principled basis for kernel design. Experiments on 36 benchmark functions (15 unimodal, 21 multimodal) across dimensions $D \in \{10, 30, 50\}$, each with 30 independent runs, show that DPSO frequently outperforms standard PSO on multimodal problems, with improvements of 2--8$\times$ on functions such as Pinter, Ackley, and Levy, and up to 5$\times$ reduction in run-to-run variance. On unimodal landscapes the modulation term is counterproductive, confirming that DPSO targets the exploration--exploitation trade-off rather than offering a universal improvement. The method adds one hyperparameter, incurs 15--25\% wall-clock overhead, and does not increase the asymptotic per-iteration complexity of PSO. The project code is available here: \href{https://github.com/Kleyt0n/dpso}{https://github.com/Kleyt0n/dpso}
\end{abstract}

\section{Introduction}

Optimization methods are essential tools in various disciplines for identifying optimal solutions to complex problems. These methods vary widely in their complexity and computational demands, making them suitable for different optimization scenarios. A common way to classify these methods is by the type and amount of information they use to guide the search for optimal solutions. For instance, gradient-based methods typically employ first-order derivatives of the objective function and, in some cases, second-order derivatives to efficiently navigate the search space. In contrast, gradient-free methods, or zero-order methods, are designed for situations where gradient information is unavailable, unreliable, or impractical to compute. These methods rely solely on function evaluations at different points in the search space to infer the direction toward the optimal solution. Among gradient-free approaches, population-based methods, such as evolutionary algorithms \citep{fogel1994evolutionary} and particle swarm optimization \citep{wang2018particle}, are particularly popular. They leverage a diverse set of candidate solutions and stochastic processes to explore the search space effectively.

Particle Swarm Optimization (PSO) introduced by \cite{kennedy1995particle} is a metaheuristic optimization algorithm inspired by the collective behavior of swarms. It is widely applied in solving high-dimensional \citep{helwig2007dimension}, nonlinear \citep{hu2002solving}, and black-box optimization problems \citep{tang2013surrogate}. The swarm dynamics were a convincing example of emergent behavior, that is, a complex global behavior arising from the interaction of simple rules or individual decisions. As pointed out by \cite{van2001analysis}, the emergency is one of the key aspects of successful results in PSO algorithm, by a simple implementation resulting in complex and accurate search dynamics. 

The standard PSO algorithm does not depend on the gradient of the objective function $f(\mathbf{x})$, making it well-suited for optimizing non-differentiable, discontinuous, or noisy functions. The algorithm employs a global model that maintains a single "best solution," referred to as the global best particle $gbest$, which represents the best-known position among all particles in the swarm. This particle acts as an attractor, influencing the movement of all other particles toward it. However, if the $gbest$ is not updated regularly, the swarm may converge prematurely, potentially becoming trapped in suboptimal minima.

Over the last three decades, numerous studies have explored and modified the PSO algorithm. For example, \cite{yildiz2009novel} proposed an algorithm that combines the particle swarm optimization technique with the receptor editing property of the immune system to address design and manufacturing problems. \cite{gang2012novel} improved the exploration-exploitation trade-off by randomly partitioning the population into sub-swarms. At periodic stages during the evolution, particles migrate between complexes to enhance population diversity. \cite{hakli2014novel} achieved promising results by integrating PSO with a specific type of random walk known as Levy flight. \cite{eschwege2024belief} designed a self-adaptive particle swarm optimization algorithm using a belief space inspired by cultural algorithms.

In this paper, we propose a novel PSO approach based on the divergence measure between the best particle position and the global best position—called Divergence-guided Particle Swarm Optimization (DPSO). The key idea behind this method is to preserve exploration throughout the iterations to identify the optimal solution to the problem. Our main findings indicate that DPSO outperforms standard PSO in several aspects: (i) DPSO demonstrates similar computing time but better accuracy compared with PSO; (ii) DPSO maintains an exploration-focused search behaviour that enhances the global solution finding. Additionally, we present an efficient parallelized and open-source algorithm\footnote{The code and experiments are available in \href{https://github.com/Kleyt0n/dpso}{this repository}.} to improve computational efficiency in resource-constrained environments. 

The document is organized as follows: Section 2 presents our novel Divergence-guided Particle Swarm Optimization metaheuristic; Section 3 shows our experiments and main results; and Section 4 describes the main findings of our work and the directions for future research.


\section{Divergence-Guided Particle Swarm Optimization}
\label{sec:dpso}

We introduce Divergence-guided Particle Swarm Optimization (DPSO), a principled extension of the standard PSO framework that mitigates premature convergence through a divergence-based velocity modulation term. The key idea is to introduce
a repulsive force that activates when a particle's search history becomes indistinguishable from the global attractor, thereby preserving exploration throughout the optimization process.

\subsection{Preliminaries}
\label{sec:preliminaries}

Consider the unconstrained minimization problem
$\min_{\mathbf{x} \in \mathbb{R}^n} f(\mathbf{x})$, where $f \colon
\mathbb{R}^n \to \mathbb{R}$ is a (possibly non-differentiable or black-box)
objective function. The standard PSO algorithm~\citep{kennedy1995particle}
maintains a swarm of $N$ particles, each characterized by a position
$\mathbf{x}_i^{(t)} \in \mathbb{R}^n$ and velocity $\mathbf{v}_i^{(t)} \in
\mathbb{R}^n$ at iteration $t$. Each particle tracks its personal best position
$\mathbf{p}_i^{(t)} \triangleq \arg\min_{\tau \le t} f(\mathbf{x}_i^{(\tau)})$,
and the swarm maintains a global best
$\mathbf{g}^{(t)} \triangleq \arg\min_{i,\,\tau \le t}
f(\mathbf{x}_i^{(\tau)})$. The dynamics follow~\citep{shi1998modified}:
\begin{align}
    \mathbf{v}_i^{(t+1)} &= \omega\, \mathbf{v}_i^{(t)}
        + c_1 r_1^{(t)} \bigl(\mathbf{p}_i^{(t)} - \mathbf{x}_i^{(t)}\bigr)
        + c_2 r_2^{(t)} \bigl(\mathbf{g}^{(t)} - \mathbf{x}_i^{(t)}\bigr),
    \label{eq:pso_vel} \\
    \mathbf{x}_i^{(t+1)} &= \mathbf{x}_i^{(t)} + \mathbf{v}_i^{(t+1)},
    \label{eq:pso_pos}
\end{align}
where $\omega > 0$ is the inertia weight, $c_1, c_2 > 0$ are the cognitive and
social acceleration coefficients, and $r_1^{(t)}, r_2^{(t)} \sim
\mathrm{Uniform}(0,1)$ are independent random scalars drawn at each iteration.

\begin{remark}[Premature convergence]
\label{rem:premature}
As the swarm evolves, personal bests $\mathbf{p}_i^{(t)}$ tend to cluster
around the global best $\mathbf{g}^{(t)}$. When
$\mathbf{p}_i^{(t)} \approx \mathbf{g}^{(t)}$ for most particles, the
cognitive and social terms in Eq. ~\eqref{eq:pso_vel} both pull toward
$\mathbf{g}^{(t)}$, collapsing the effective search radius. If
$\mathbf{g}^{(t)}$ is a local (but not global) minimizer, the swarm becomes
trapped. This failure mode is well-documented in the
literature~\citep{van2001analysis} and motivates our approach.
\end{remark}

\subsection{The DPSO velocity update}
\label{sec:dpso_update}

DPSO augments (Eq. \eqref{eq:pso_vel}) with an additive modulation term that applies a
repulsive force pushing particles away from $\mathbf{g}^{(t)}$, with magnitude
governed by the proximity of $\mathbf{p}_i^{(t)}$ to $\mathbf{g}^{(t)}$:
\begin{equation}
\label{eq:dpso_vel}
\boxed{
    \mathbf{v}_i^{(t+1)} = \underbrace{\omega\, \mathbf{v}_i^{(t)}}_{\text{inertia}}
        + \underbrace{c_1 r_1^{(t)} \bigl(\mathbf{p}_i^{(t)} - \mathbf{x}_i^{(t)}\bigr)}_{\text{cognitive}}
        + \underbrace{c_2 r_2^{(t)} \bigl(\mathbf{g}^{(t)} - \mathbf{x}_i^{(t)}\bigr)}_{\text{social}}
        + \underbrace{\mathbf{v}_{\mathrm{mod},i}^{(t)}}_{\text{modulation}}
}
\end{equation}
with position update $\mathbf{x}_i^{(t+1)} = \mathbf{x}_i^{(t)} +
\mathbf{v}_i^{(t+1)}$ as before. The modulation term is defined as follows.

\begin{definition}[Divergence modulation]
\label{def:vmod}
Let $c_3 \ge 0$ be a modulation strength coefficient,
$r_3^{(t)} \sim \mathrm{Uniform}(0,1)$, and $\sigma > 0$ a bandwidth
parameter. The modulation velocity for particle $i$ at iteration $t$ is:
\begin{equation}
\label{eq:vmod}
    \mathbf{v}_{\mathrm{mod},i}^{(t)}
        = c_3\, r_3^{(t)}\;
          \kappa\!\bigl(\mathbf{p}_i^{(t)},\, \mathbf{g}^{(t)}\bigr)
          \;\hat{\mathbf{d}}_i^{(t)},
\end{equation}
where $\kappa \colon \mathbb{R}^n \times \mathbb{R}^n \to [0,1]$ is a
\emph{similarity kernel},
\begin{equation}
\label{eq:kernel}
    \kappa\!\bigl(\mathbf{p}_i^{(t)},\, \mathbf{g}^{(t)}\bigr)
        = \exp\!\left(
            -\frac{\bigl\|\mathbf{p}_i^{(t)} - \mathbf{g}^{(t)}\bigr\|_2^2}
                  {2\sigma^2}
          \right),
\end{equation}
and $\hat{\mathbf{d}}_i^{(t)} \in \mathbb{R}^n$ is the \emph{repulsion
direction},
\begin{equation}
\label{eq:repulsion_dir}
    \hat{\mathbf{d}}_i^{(t)}
        = \frac{\mathbf{x}_i^{(t)} - \mathbf{g}^{(t)}}
               {\bigl\|\mathbf{x}_i^{(t)} - \mathbf{g}^{(t)}\bigr\|_2
                + \epsilon},
\end{equation}
with $\epsilon > 0$ a small constant (e.g., $\epsilon = 10^{-9}$) for numerical
stability.
\end{definition}

The design of Eq. ~\eqref{eq:vmod} rests on two interacting components whose roles we
now clarify.

\paragraph{Similarity kernel (activation gate).}
The Gaussian kernel $\kappa$ acts as a smooth gate that controls \emph{when}
repulsion activates. It satisfies $\kappa \to 1$ as
$\mathbf{p}_i^{(t)} \to \mathbf{g}^{(t)}$ and decays to zero as the distance
$\|\mathbf{p}_i^{(t)} - \mathbf{g}^{(t)}\|_2$ grows relative to $\sigma$.
Consequently, the modulation term has appreciable magnitude only for particles
whose personal best has converged near the current global best—precisely the
particles at risk of premature convergence (cf.\ Remark~\ref{rem:premature}).
The bandwidth $\sigma$ controls the activation radius: smaller $\sigma$ restricts
repulsion to particles whose personal bests are very close to $\mathbf{g}^{(t)}$,
while larger $\sigma$ applies gentler repulsion over a wider neighborhood.

\paragraph{Repulsion direction.}
The unit vector $\hat{\mathbf{d}}_i^{(t)}$ points from $\mathbf{g}^{(t)}$
toward the particle's \emph{current} position $\mathbf{x}_i^{(t)}$, not its
personal best. This choice is deliberate: it pushes the particle along its
existing displacement from the global attractor, encouraging exploration of
regions the particle already occupies rather than creating arbitrary displacements.
Combined with the stochastic factor $r_3^{(t)}$, this produces diversified
trajectories without requiring an explicit randomization of direction.

\begin{remark}[Interpolation between PSO and DPSO]
\label{rem:interpolation}
Setting $c_3 = 0$ recovers standard PSO exactly. For $c_3 > 0$, the modulation
strength is bounded: $\|\mathbf{v}_{\mathrm{mod},i}^{(t)}\|_2 \le c_3$ almost
surely, since $\kappa \le 1$, $r_3^{(t)} \le 1$, and
$\|\hat{\mathbf{d}}_i^{(t)}\|_2 \le 1$. This ensures the modulation term does
not dominate the velocity update regardless of the choice of $c_3$.
\end{remark}

\subsection{Connection to $f$-divergences}
\label{sec:fdiv}

The Gaussian kernel in Eq. \eqref{eq:kernel} admits a natural interpretation through
the lens of $f$-divergences, which provides a principled framework for
alternative modulation designs.

Recall that for probability distributions $P$ and $Q$ with densities $p$ and
$q$ on $\mathcal{X} \subseteq \mathbb{R}^n$, the $f$-divergence is
\begin{equation}
\label{eq:fdiv_def}
    D_f(P \,\|\, Q) = \int_{\mathcal{X}} q(\mathbf{x})\,
        f\!\left(\frac{p(\mathbf{x})}{q(\mathbf{x})}\right) \mathrm{d}\mathbf{x},
\end{equation}
where $f \colon (0,\infty) \to \mathbb{R}$ is convex with $f(1) = 0$.

To relate point estimates to distributions, we embed $\mathbf{p}_i^{(t)}$ and
$\mathbf{g}^{(t)}$ as the means of isotropic Gaussians with shared bandwidth
$\sigma_k > 0$:
\begin{equation}
\label{eq:gaussian_embed}
    P_i \sim \mathcal{N}\!\bigl(\mathbf{p}_i^{(t)},\, \sigma_k^2 \mathbf{I}_n\bigr),
    \qquad
    Q_g \sim \mathcal{N}\!\bigl(\mathbf{g}^{(t)},\, \sigma_k^2 \mathbf{I}_n\bigr).
\end{equation}
For the Kullback--Leibler divergence ($f(u) = u \log u$), the divergence between
two isotropic Gaussians with equal covariance simplifies to~\citep{duchi2007derivations}:
\begin{equation}
\label{eq:kl_closed}
    D_{\mathrm{KL}}(P_i \,\|\, Q_g)
        = \frac{\bigl\|\mathbf{p}_i^{(t)} - \mathbf{g}^{(t)}\bigr\|_2^2}
               {2\sigma_k^2}.
\end{equation}

\begin{proposition}[Equivalence to KL-based modulation]
\label{prop:kl_equiv}
The modulation term (Equation \eqref{eq:vmod}) with kernel (Equation \eqref{eq:kernel}) is equivalent
to
\begin{equation}
\label{eq:vmod_kl}
    \mathbf{v}_{\mathrm{mod},i}^{(t)}
        = c_3\, r_3^{(t)}\;
          \exp\!\bigl(-\alpha\, D_{\mathrm{KL}}(P_i \,\|\, Q_g)\bigr)
          \;\hat{\mathbf{d}}_i^{(t)},
\end{equation}
with $\alpha = \sigma_k^2 / \sigma^2$.
\end{proposition}

\begin{proof}
Substituting Eq. \eqref{eq:kl_closed} into Eq. \eqref{eq:vmod_kl}:

\begin{equation}
\begin{array}{rl}    
    \exp\!\bigl(-\alpha\, D_{\mathrm{KL}}(P_i \,\|\, Q_g)\bigr)
    = \exp\!\left(
        -\frac{\sigma_k^2}{\sigma^2}
        \cdot \frac{\bigl\|\mathbf{p}_i^{(t)} - \mathbf{g}^{(t)}\bigr\|_2^2}
                   {2\sigma_k^2}
      \right)
    &\\ = \exp\!\left(
        -\frac{\bigl\|\mathbf{p}_i^{(t)} - \mathbf{g}^{(t)}\bigr\|_2^2}
              {2\sigma^2}
      \right)
    &\\ =  \kappa\!\bigl(\mathbf{p}_i^{(t)}, \mathbf{g}^{(t)}\bigr). \qedhere
\end{array}
\end{equation}
\end{proof}

This equivalence reveals that the Gaussian kernel is a monotonically decreasing
function of the KL divergence between the Gaussian-embedded personal and global
bests: repulsion is strongest when the two distributions overlap most
(low divergence) and vanishes as they separate (high divergence). We refer to
this design as \emph{repulsion proportional to similarity} (RPS).

\begin{remark}[Generalization via alternative $f$-divergences]
\label{rem:generalization}
Proposition~\ref{prop:kl_equiv} extends to any $f$-divergence that admits a
closed-form expression under the Gaussian embedding (Eq. \eqref{eq:gaussian_embed}).
Replacing $D_{\mathrm{KL}}$ in Eq. \eqref{eq:vmod_kl} with, for example, the squared
Hellinger distance ($f(u) = (\sqrt{u} - 1)^2$), the Jensen--Shannon divergence,
or the total variation distance ($f(u) = \tfrac{1}{2}|u-1|$) yields alternative
kernels with different sensitivity profiles. For instance, the Hellinger distance
between equal-covariance Gaussians is
$H^2(P_i, Q_g) = 1 - \exp\!\bigl(-\|\mathbf{p}_i^{(t)} -
\mathbf{g}^{(t)}\|_2^2 / (8\sigma_k^2)\bigr)$, which saturates more rapidly
than the KL divergence and may be preferable when a sharper activation boundary
is desired. The choice of divergence thus provides a design knob for shaping the
exploration--exploitation trade-off.
\end{remark}

\subsection{Algorithm}
\label{sec:algorithm}

Algorithm~\ref{alg:dpso} presents the complete DPSO procedure. The modulation
term integrates seamlessly into the standard PSO loop: lines~\ref{line:dir}--\ref{line:vmod} are the only additions, and each particle's
update remains independent, preserving the suitability for parallelization.

\begin{algorithm}[H]
\caption{Divergence-Guided Particle Swarm Optimization (DPSO)}
\label{alg:dpso}
\begin{algorithmic}[1]
\REQUIRE Objective $f$, swarm size $N$, dimension $n$, coefficients
    $\omega, c_1, c_2, c_3$, bandwidth $\sigma$, max.\ iterations $T$,
    search bounds $[\mathbf{lb}, \mathbf{ub}] \subset \mathbb{R}^n$.
\STATE Initialize $\mathbf{x}_i^{(0)} \sim \mathrm{Uniform}(\mathbf{lb},
    \mathbf{ub})$, $\mathbf{v}_i^{(0)} = \mathbf{0}$, $\mathbf{p}_i^{(0)}
    \gets \mathbf{x}_i^{(0)}$ for $i = 1, \dots, N$.
\STATE $\mathbf{g}^{(0)} \gets \arg\min_{i} f(\mathbf{p}_i^{(0)})$.
\FOR{$t = 0, 1, \dots, T-1$}
    \FOR{$i = 1, \dots, N$ \textbf{in parallel}}
        \STATE Draw $r_1, r_2, r_3 \sim \mathrm{Uniform}(0, 1)$ independently.
        \STATE $\hat{\mathbf{d}}_i^{(t)} \gets
            \bigl(\mathbf{x}_i^{(t)} - \mathbf{g}^{(t)}\bigr) \big/
            \bigl(\|\mathbf{x}_i^{(t)} - \mathbf{g}^{(t)}\|_2 + \epsilon\bigr)$
            \label{line:dir}
            \hfill $\triangleright$ Repulsion direction
        \STATE $\kappa_i \gets \exp\!\bigl(-\|\mathbf{p}_i^{(t)} -
            \mathbf{g}^{(t)}\|_2^2 \big/ (2\sigma^2)\bigr)$
            \hfill $\triangleright$ Similarity kernel
        \STATE $\mathbf{v}_{\mathrm{mod},i}^{(t)} \gets c_3\, r_3\,
            \kappa_i\, \hat{\mathbf{d}}_i^{(t)}$
            \label{line:vmod}
            \hfill $\triangleright$ Divergence modulation
        \STATE $\mathbf{v}_i^{(t+1)} \gets \omega\, \mathbf{v}_i^{(t)}
            + c_1 r_1 (\mathbf{p}_i^{(t)} - \mathbf{x}_i^{(t)})
            + c_2 r_2 (\mathbf{g}^{(t)} - \mathbf{x}_i^{(t)})
            + \mathbf{v}_{\mathrm{mod},i}^{(t)}$
        \STATE $\mathbf{v}_i^{(t+1)} \gets
            \mathrm{clamp}\bigl(\mathbf{v}_i^{(t+1)},\, -\mathbf{v}_{\max},\,
            \mathbf{v}_{\max}\bigr)$
            \hfill $\triangleright$ Velocity clamping
        \STATE $\mathbf{x}_i^{(t+1)} \gets
            \mathrm{clamp}\bigl(\mathbf{x}_i^{(t)} + \mathbf{v}_i^{(t+1)},\,
            \mathbf{lb},\, \mathbf{ub}\bigr)$
            \hfill $\triangleright$ Boundary enforcement
        \STATE Evaluate $f_i \gets f(\mathbf{x}_i^{(t+1)})$.
        \IF{$f_i < f(\mathbf{p}_i^{(t)})$}
            \STATE $\mathbf{p}_i^{(t+1)} \gets \mathbf{x}_i^{(t+1)}$
        \ELSE
            \STATE $\mathbf{p}_i^{(t+1)} \gets \mathbf{p}_i^{(t)}$
        \ENDIF
    \ENDFOR
    \STATE $\mathbf{g}^{(t+1)} \gets \arg\min\bigl\{f(\mathbf{g}^{(t)}),\,
        \min_i f(\mathbf{p}_i^{(t+1)})\bigr\}$
        \hfill $\triangleright$ Synchronous global best update
\ENDFOR
\RETURN $\mathbf{g}^{(T)},\; f(\mathbf{g}^{(T)})$.
\end{algorithmic}
\end{algorithm}

\subsection{Computational complexity}
\label{sec:complexity}

Let $C_f$ denote the cost of a single evaluation of $f$.

\begin{proposition}[Complexity]
\label{prop:complexity}
The per-iteration cost of DPSO is $\Theta\bigl(N(n + C_f)\bigr)$, identical in
order to standard PSO. Over $T$ iterations, the total cost is
$\Theta\bigl(T \cdot N \cdot (n + C_f)\bigr)$.
\end{proposition}

\begin{proof}[Proof sketch]
The modulation term adds three $O(n)$ operations per particle per iteration:
(i)~the norm $\|\mathbf{p}_i^{(t)} - \mathbf{g}^{(t)}\|_2$ for the kernel,
(ii)~the normalized direction $\hat{\mathbf{d}}_i^{(t)}$, and
(iii)~the scalar--vector product $c_3 r_3 \kappa_i \hat{\mathbf{d}}_i^{(t)}$.
Since standard PSO already performs $\Theta(n)$ work per particle (three vector
additions/subtractions), the modulation increases the constant in the $O(Nn)$
term but does not change the asymptotic complexity.
\end{proof}

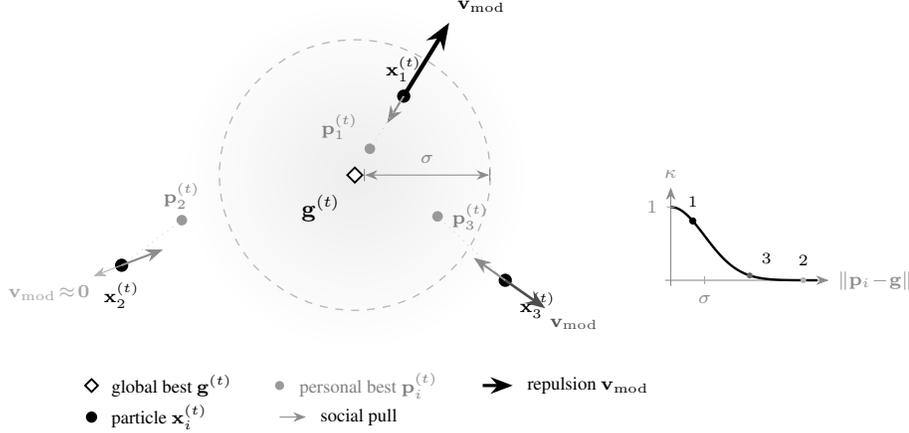
\begin{figure}[H]
\begin{tikzpicture}[
    >=Stealth,
    every node/.style={font=\small},
    particle/.style={circle, fill=black, inner sep=0pt, minimum size=5pt},
    gbest/.style={diamond, draw=black, fill=white, inner sep=0pt, minimum size=6pt, line width=0.7pt},
    pbest/.style={circle, fill=black!40, inner sep=0pt, minimum size=4pt},
    arrsocial/.style={->, thick, color=black!45},
    arrrepstrong/.style={->, line width=1.6pt, color=black},
    arrrepmod/.style={->, line width=1.0pt, color=black!70},
    arrrepweak/.style={->, thin, color=black!30},
    lbl/.style={font=\scriptsize},
]

\shade[inner color=black!7, outer color=white] (0,0) circle (2.3);
\draw[black!28, dashed, line width=0.5pt] (0,0) circle (1.8);

\draw[black!45, thin, |<->|] (0.12,0) -- (1.8,0);
\node[lbl, text=black!50, anchor=south] at (0.96,0.04) {$\sigma$};

\node[gbest] (g) at (0, 0) {};
\node[font=\small, anchor=north east, xshift=-2pt, yshift=-4pt] at (g) {$\mathbf{g}^{(t)}$};

\node[particle] (x1) at (0.65, 1.05) {};
\node[lbl, anchor=south, yshift=2pt] at (x1) {$\mathbf{x}_1^{(t)}$};

\node[pbest] (p1) at (0.2, 0.35) {};
\node[lbl, anchor=south east, xshift=-1pt, text=black!50] at (p1) {$\mathbf{p}_1^{(t)}$};

\draw[arrsocial] (x1.center) -- +(-0.22,-0.36);

\pgfmathsetmacro{\nAx}{0.65}
\pgfmathsetmacro{\nAy}{1.05}
\pgfmathsetmacro{\nA}{sqrt(\nAx*\nAx+\nAy*\nAy)}
\draw[arrrepstrong] (x1.center) -- +({1.15*\nAx/\nA},{1.15*\nAy/\nA})
    node[pos=1, lbl, anchor=south west, xshift=-1pt] {$\mathbf{v}_{\mathrm{mod}}$};

\node[particle] (x2) at (-3.1, -1.2) {};
\node[lbl, anchor=north, yshift=-3pt] at (x2) {$\mathbf{x}_2^{(t)}$};

\node[pbest] (p2) at (-2.3, -0.6) {};
\node[lbl, anchor=south, yshift=1pt, text=black!50] at (p2) {$\mathbf{p}_2^{(t)}$};

\draw[arrsocial] (x2.center) -- +(0.55, 0.21);

\pgfmathsetmacro{\nBx}{-3.1}
\pgfmathsetmacro{\nBy}{-1.2}
\pgfmathsetmacro{\nB}{sqrt(\nBx*\nBx+\nBy*\nBy)}
\draw[arrrepweak] (x2.center) -- +({0.4*\nBx/\nB},{0.4*\nBy/\nB})
    node[pos=1, lbl, anchor=north east, text=black!35, xshift=2pt] {$\mathbf{v}_{\mathrm{mod}}\!\approx\!\mathbf{0}$};

\node[particle] (x3) at (2.0, -1.4) {};
\node[lbl, anchor=north west, xshift=2pt, yshift=-1pt] at (x3) {$\mathbf{x}_3^{(t)}$};

\node[pbest] (p3) at (1.1, -0.55) {};
\node[lbl, anchor=west, xshift=2pt, text=black!50] at (p3) {$\mathbf{p}_3^{(t)}$};

\draw[arrsocial] (x3.center) -- +(-0.42, 0.30);

\pgfmathsetmacro{\nCx}{2.0}
\pgfmathsetmacro{\nCy}{-1.4}
\pgfmathsetmacro{\nC}{sqrt(\nCx*\nCx+\nCy*\nCy)}
\draw[arrrepmod] (x3.center) -- +({0.65*\nCx/\nC},{0.65*\nCy/\nC})
    node[pos=1, lbl, anchor=north west, xshift=-2pt] {$\mathbf{v}_{\mathrm{mod}}$};

\draw[black!18, dotted, thin] (x1.center) -- (p1.center);
\draw[black!18, dotted, thin] (x2.center) -- (p2.center);
\draw[black!18, dotted, thin] (x3.center) -- (p3.center);

\begin{scope}[shift={(4.2,-1.4)}, scale=0.75]
    \draw[->, black!45, thin] (-0.1,0) -- (2.8,0)
        node[right, font=\scriptsize, text=black!45] {$\|\mathbf{p}_i\!-\!\mathbf{g}\|$};
    \draw[->, black!45, thin] (0,-0.1) -- (0,1.65)
        node[above, font=\scriptsize, text=black!45] {$\kappa$};

    \draw[black, line width=0.9pt, domain=0:2.6, samples=80]
        plot (\x, {1.3*exp(-(\x*\x)/(2*0.6*0.6))});

    \draw[black!35, thin] (-0.07, 1.3) -- (0.07, 1.3);
    \node[font=\scriptsize, text=black!45, anchor=east] at (-0.07, 1.3) {$1$};

    \draw[black!35, thin] (0.6, -0.07) -- (0.6, 0.07);
    \node[font=\scriptsize, text=black!45, anchor=north] at (0.6, -0.1) {$\sigma$};

    \fill[black] ({0.39}, {1.3*exp(-(0.39*0.39)/(2*0.6*0.6))}) circle (1.8pt);
    \fill[black!60] ({1.4}, {1.3*exp(-(1.4*1.4)/(2*0.6*0.6))}) circle (1.5pt);
    \fill[black!30] ({2.35}, {1.3*exp(-(2.35*2.35)/(2*0.6*0.6))}) circle (1.5pt);

    \node[font=\tiny, anchor=south, yshift=2pt] at ({0.39}, {1.3*exp(-(0.39*0.39)/(2*0.6*0.6))}) {$1$};
    \node[font=\tiny, anchor=south west, xshift=1pt, yshift=1pt] at ({1.4}, {1.3*exp(-(1.4*1.4)/(2*0.6*0.6))}) {$3$};
    \node[font=\tiny, anchor=south, yshift=2pt] at ({2.35}, {1.3*exp(-(2.35*2.35)/(2*0.6*0.6))}) {$2$};
\end{scope}

\begin{scope}[shift={(-3.5, -2.8)}]
    \node[gbest, scale=0.85] at (0,0) {};
    \node[lbl, anchor=west, xshift=4pt] at (0,0) {global best $\mathbf{g}^{(t)}$};

    \node[particle, scale=0.85] at (0,-0.42) {};
    \node[lbl, anchor=west, xshift=4pt] at (0,-0.42) {particle $\mathbf{x}_i^{(t)}$};

    \node[pbest, scale=0.85] at (2.5,0) {};
    \node[lbl, anchor=west, xshift=4pt, text=black!50] at (2.5,0) {personal best $\mathbf{p}_i^{(t)}$};

    \draw[arrsocial, thin] (2.5,-0.42) -- +(0.35,0);
    \node[lbl, anchor=west, xshift=2pt] at (2.85,-0.42) {social pull};

    \draw[arrrepstrong, line width=1.2pt] (5.2, 0) -- +(0.4,0);
    \node[lbl, anchor=west, xshift=2pt] at (5.6, 0) {repulsion $\mathbf{v}_{\mathrm{mod}}$};
\end{scope}
\end{tikzpicture}
\caption{Illustration of the DPSO divergence modulation mechanism in a two-dimensional search space.}
\label{fig:diagram}
\end{figure}

In Figure \ref{fig:diagram}, the shaded region around the global best $\mathbf{g}^{(t)}$ (dashed circle, radius $\sigma$) indicates the activation zone of the similarity kernel $\kappa = \exp\!\bigl({-\|\mathbf{p}_i^{(t)}-\mathbf{g}^{(t)}\|^2}/{2\sigma^2}\bigr)$. Three particles illustrate the distance-dependent behaviour: particle~1, whose personal best $\mathbf{p}_1^{(t)}$ lies close to $\mathbf{g}^{(t)}$ ($\kappa \approx 1$), receives a strong repulsive velocity $\mathbf{v}_{\mathrm{mod}}$ directed away from $\mathbf{g}^{(t)}$; particle~3, with $\mathbf{p}_3^{(t)}$ near the activation boundary, receives moderate repulsion; particle~2, whose $\mathbf{p}_2^{(t)}$ is well outside the activation zone ($\kappa \approx 0$), is effectively governed by standard PSO dynamics alone. Gray arrows denote the social attraction component toward $\mathbf{g}^{(t)}$. \textbf{(Inset)}~The kernel profile $\kappa$ as a function of $\|\mathbf{p}_i - \mathbf{g}\|$, with the positions of particles 1-3 marked, showing the exponential decay that governs the modulation strength.

\begin{remark}[Practical overhead]
When $C_f \gg n$ (e.g., simulation-based objectives), the overhead of the modulation term is negligible. When $C_f = O(n)$ (e.g., quadratic benchmarks), we observe empirically that DPSO incurs ${\sim}15\text{--}25\%$ additional wall-clock time per iteration (see Section~\ref{sec:experiments}), which is offset by improved solution quality and reduced iteration counts to convergence.
\end{remark}

\subsection{Discussion of design choices}
\label{sec:design}

\paragraph{Why repulsion proportional to similarity?}

An alternative design repulsion proportional to \textit{dissimilarity}—would scale the modulation as an increasing function of
$\|\mathbf{p}_i^{(t)} - \mathbf{g}^{(t)}\|_2$, e.g., $S_f = \|\mathbf{p}_i^{(t)} - \mathbf{g}^{(t)}\|_2^2 / (2\sigma_k^2)$.
We argue this is undesirable: particles whose personal bests already differ substantially from the global best are already exploring diverse regions. Adding stronger repulsion to these particles risks destabilizing their search and wasting function evaluations on increasingly remote regions. The RPS design targets intervention precisely where it is needed—at particles converging toward $\mathbf{g}^{(t)}$—while leaving well-separated particles undisturbed.

\paragraph{Choice of $\sigma$.}

The bandwidth $\sigma$ should reflect the scale of the search space. In our experiments, we set $\sigma = \beta \cdot \|\mathbf{ub} - \mathbf{lb}\|_2$ with $\beta \in [0.05, 0.2]$, ensuring that the activation region is a fixed fraction of the domain diameter. Adaptive schedules for $\sigma$ (e.g., decaying with iteration count) are a natural extension but are not explored in this work.

\paragraph{Interaction with inertia decay.}

Standard practice often decays $\omega$ linearly from $\omega_{\max}$ to $\omega_{\min}$ over the iteration budget~\citep{shi1998modified}. DPSO is compatible with this schedule: the modulation term provides an additional exploration mechanism that compensates for reduced inertia in later iterations, preventing the swarm from stalling even when $\omega$ is small.

\section{Experiments and Discussion}
\label{sec:experiments}

We evaluate DPSO against standard PSO on 36 benchmark functions drawn from \cite{jamil2013literature}, varying the problem dimensionality to assess scalability. All experiments use identical swarm
configurations and random seeds, isolating the effect of the divergence
modulation term.

\subsection{Experimental setup}
\label{sec:setup}

\paragraph{Benchmark functions.}
Table~\ref{tab:benchmarks} lists the 36 test functions used in our evaluation.
The suite comprises 15 unimodal functions (e.g., Sphere, Rosenbrock, Schwefel
variants, Zakharov, Cigar) and 21 multimodal functions (e.g., Rastrigin, Ackley,
Levy, Weierstrass, HappyCat), providing a comprehensive test of the
exploration--exploitation trade-off across a wide range of landscape
characteristics. See Appendix \ref{app:functions} for additional information on these benchmark functions.

\paragraph{Algorithm configuration.}
Both PSO and DPSO use $N = 40$ particles, $T = 1{,}000$ iterations, and
Clerc's constriction coefficients~\citep{clerc2002particle}: $\omega = 0.7298$,
$c_1 = c_2 = 1.49618$. Velocities are clamped to
$v_{\max} = 0.2 \cdot (\mathbf{ub} - \mathbf{lb})$ per dimension, and
positions are clipped to the search bounds. For DPSO, we set $c_3 = 1.0$ and
$\beta = 0.1$ (so $\sigma = 0.1 \cdot \|\mathbf{ub} - \mathbf{lb}\|_2$).

\paragraph{Protocol.}
Each algorithm--function--dimension combination is run $30$ times with
independent random seeds (master seed $= 42$). We test dimensions
$D \in \{10, 30, 50\}$. All experiments are implemented in JAX with
\texttt{jax.lax.scan} for JIT-compiled iteration loops and
\texttt{jax.vmap} for vectorized fitness evaluation. Wall-clock times are
measured with \texttt{jax.block\_until\_ready} to ensure accurate timing.

\subsection{Results}
\label{sec:results}

\paragraph{Solution quality.}
Tables~\ref{tab:results-unimodal} and~\ref{tab:results-multimodal} report the best fitness achieved by each algorithm
(mean $\pm$ standard deviation over 30 runs). DPSO outperforms PSO on the
majority of multimodal benchmark--dimension combinations, with the most
substantial improvements in higher dimensions.

\input{outputs/tables/results_unimodal.tex}
\input{outputs/tables/results_multimodal.tex}

The largest single multimodal improvement is on Pinter at $D = 10$, where
DPSO achieves an $8.4\times$ reduction in mean fitness ($3.88$ vs.\ $32.5$),
indicating that the modulation term effectively prevents stagnation in
Pinter's many narrow basins. On Ackley, DPSO achieves a $2.8\times$
improvement at $D = 30$ (mean $4.34 \times 10^{-1}$ vs.\ $1.20$) and
$3.6\times$ at $D = 50$ ($8.98 \times 10^{-1}$ vs.\ $3.27$). Further
multimodal wins include Levy at $D = 30$ ($2.6\times$; $1.60$ vs.\ $4.18$)
and Salomon, where DPSO consistently improves by $1.2$--$1.4\times$ across
all dimensions. On Rastrigin, DPSO improves by $38\%$ at $D = 10$ and
$10\%$ at $D = 30$, though PSO recovers the advantage at $D = 50$ (where
DPSO is $9.5\%$ worse). On Griewank at $D = 50$, DPSO reduces the mean
fitness by $1.7\times$ ($7.02 \times 10^{-2}$ vs.\ $1.19 \times 10^{-1}$).
Among nominally unimodal functions, Rosenbrock at $D = 50$ shows a
$2.7\times$ improvement ($4{,}816$ vs.\ $12{,}835$), where DPSO's
exploration helps escape the narrow curved valley that traps standard PSO at
high dimension.

Conversely, PSO dominates on the unimodal Sphere function across all
dimensions, converging to near-machine-precision solutions
($\sim 10^{-19}$ at $D = 30$) while DPSO's modulation term introduces
residual exploration that prevents exact convergence (mean $\sim 10^{-1}$).
This is expected: on convex landscapes with a single basin of attraction,
any repulsive force is counterproductive. PSO also wins on several
multimodal functions whose basin structure is relatively simple---including
Weierstrass, Whitley, Exponential, Qing, Alpine1, CosineMixture, and
StretchedV---and marginally outperforms DPSO on Griewank at $D = 30$ and
Schwefel at $D = 30$, where the standard algorithm's convergence pressure
suffices to locate good solutions.

\paragraph{Convergence dynamics.}
Figures~\ref{fig:convergence_unimodal}--\ref{fig:convergence_multimodal} show
the convergence curves (median and interquartile range over 30 runs) for all
benchmarks across dimensions. The plots reveal two distinct regimes:

\begin{enumerate}
    \item \textbf{Early iterations} ($t < 200$): Both algorithms converge at
    similar rates, as the modulation term has little effect when personal bests
    are spread across the search space ($\kappa \approx 0$ for most particles);

    \item \textbf{Late iterations} ($t > 200$): Standard PSO plateaus as the
    swarm collapses around $\mathbf{g}^{(t)}$, while DPSO continues to improve by activating the repulsion mechanism on converged particles. This
    effect is most pronounced on Ackley and Rastrigin, where the interquartile range (IQR)  band for DPSO remains tight and the median continues to decrease, indicating consistent progress across runs.
\end{enumerate}

\paragraph{Computational overhead.}
Table~\ref{tab:timing} reports the mean wall-clock time per run. DPSO incurs
a modest overhead of approximately $15$--$25\%$ relative to PSO, consistent
with the $O(n)$ additional operations per particle per iteration
(Proposition~\ref{prop:complexity}). For all tested configurations, a single
run completes in under $0.5$ seconds on a standard CPU, confirming that the
modulation term does not impose a meaningful computational burden.

\input{outputs/tables/timing.tex}

\subsection{Discussion}
\label{sec:discussion}

The experimental results support the central thesis of this work: the
divergence modulation term preserves exploration in regions of the search
space where standard PSO's dynamics lead to premature convergence, while
adding negligible computational cost.

\paragraph{Dimension scaling.}
A consistent pattern emerges across many multimodal functions: DPSO's
advantage grows with dimensionality. At $D = 10$, the search space is small
enough that standard PSO's random perturbations suffice to avoid most local
minima. At $D = 30$ and $D = 50$, the exponential growth in the number of
local optima makes the systematic exploration provided by the modulation
term increasingly valuable. This trend is visible on Ackley, where DPSO's
improvement rises from $2.8\times$ at $D = 30$ to $3.6\times$ at $D = 50$,
and on Griewank ($1.2\times$ at $D = 10$ to $1.7\times$ at $D = 50$). On
Pinter, the improvement is $8.4\times$ at $D = 10$, where PSO easily
stagnates, and remains $1.3\times$ even at $D = 50$.

\begin{figure}[t]
    \centering
    \includegraphics[width=\textwidth]{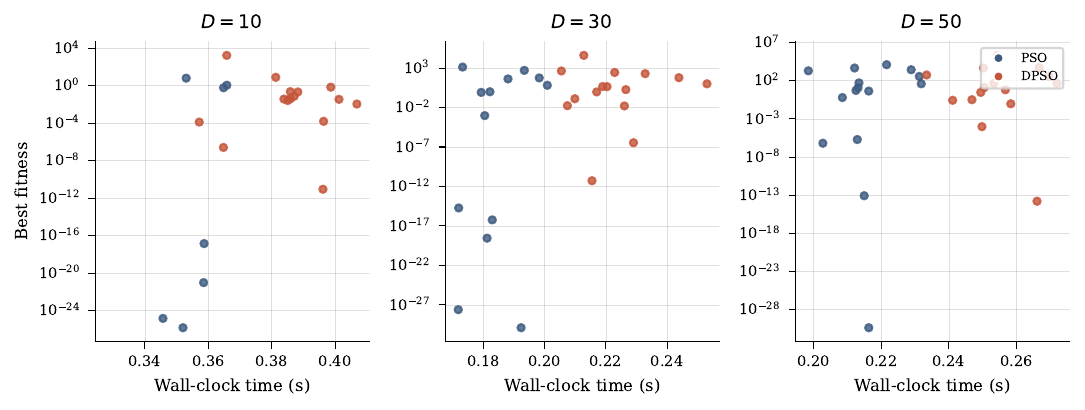}
    \caption{Fitness vs.\ wall-clock time for \textbf{unimodal} benchmarks.
    Each function produces two points: one for {\color{cRPNI}PSO} and one for
    {\color{cOurs}DPSO}. The $y$-axis shows best fitness (log scale) and
    the $x$-axis shows mean wall-clock time (seconds) over 30 runs.}
    \label{fig:tradeoff-unimodal}
\end{figure}

\begin{figure}[t]
    \centering
    \includegraphics[width=\textwidth]{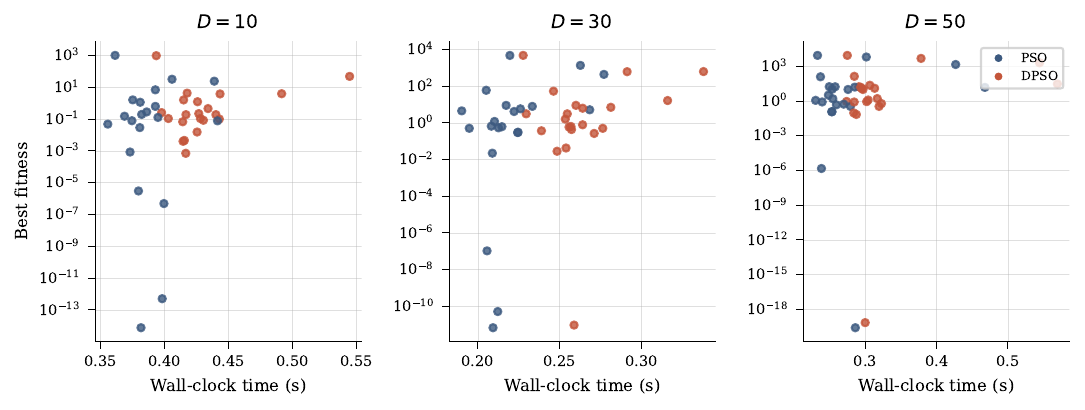}
    \caption{Fitness vs.\ wall-clock time for \textbf{multimodal} benchmarks.
    Layout as in Fig.~\ref{fig:tradeoff-unimodal}.}
    \label{fig:tradeoff-multimodal}
\end{figure}

\paragraph{The unimodal trade-off.}
As seen in Fig.~\ref{fig:tradeoff-unimodal}, the Sphere results illustrate an inherent trade-off: on convex landscapes,
any diversification mechanism slows convergence to the global optimum. The
modulation term, by construction, pushes particles away from
$\mathbf{g}^{(t)}$ when their personal bests converge—precisely the
behaviour that accelerates convergence on unimodal functions. This suggests
that an adaptive schedule for $c_3$ (e.g., decaying with a stagnation
metric) could recover PSO's convergence speed on easy landscapes while
retaining DPSO's exploration benefits on harder ones. We leave this extension
to future work.

\paragraph{Variance reduction.}
Beyond mean performance, DPSO consistently exhibits lower variance across
runs on the functions where it outperforms PSO. For example, on Pinter at
$D = 10$, PSO's standard deviation is $42.1$ while DPSO's is $10.5$ (a
$4\times$ reduction), and on Ackley at $D = 50$, PSO's standard deviation
is $0.93$ versus DPSO's $0.36$ ($2.6\times$ reduction). On Griewank at
$D = 50$, PSO's standard deviation is $0.32$ while DPSO's is $0.067$
(${\approx}5\times$ lower). This reliability is a practical advantage in
applications where a single optimization run must produce a good solution.

\section{Conclusions}
\label{sec:conclusion}

We introduced Divergence-guided Particle Swarm Optimization (DPSO), which
augments the standard PSO velocity update with a modulation term that applies
repulsive forces to particles whose personal bests have converged near the
global best. The modulation is governed by a Gaussian similarity kernel,
which we showed is equivalent to an exponentially decaying function of the
KL divergence between Gaussian-embedded personal and global bests. This
connection to $f$-divergences provides a principled framework for designing
alternative kernels with different exploration profiles.

Experiments on 36 benchmark functions across dimensions $D \in \{10, 30, 50\}$
demonstrated that DPSO frequently outperforms standard PSO on multimodal
landscapes, with improvements that often grow with dimensionality.
DPSO also exhibited lower variance across runs on the functions where it
improves mean performance, indicating more reliable convergence. The
computational overhead was modest ($15$--$25\%$), and the method required
only one additional hyperparameter ($c_3$) beyond standard PSO.

Future work includes adaptive scheduling of $c_3$ and $\sigma$ to balance
exploration and exploitation across the optimization trajectory, evaluation
on real-world black-box optimization problems, and investigation of
alternative $f$-divergence kernels (e.g., Hellinger, Jensen--Shannon) as
discussed in Remark~\ref{rem:generalization}.

\section{Acknowledgments}

The author thanks Dr. Georges Spyrides for his comments and feedback on this work. Any remaining errors are the author's own.

\bibliographystyle{abbrvnat}
\bibliography{references}

\appendix
\section{Convergence Plots}
\label{app:convergence}

The following figures show convergence curves for all remaining benchmark
functions not included in the main text. Layout is identical to
Figures~\ref{fig:convergence_unimodal}--\ref{fig:convergence_multimodal}.

\begin{figure}[H]
    \centering
    \includegraphics[width=\textwidth]{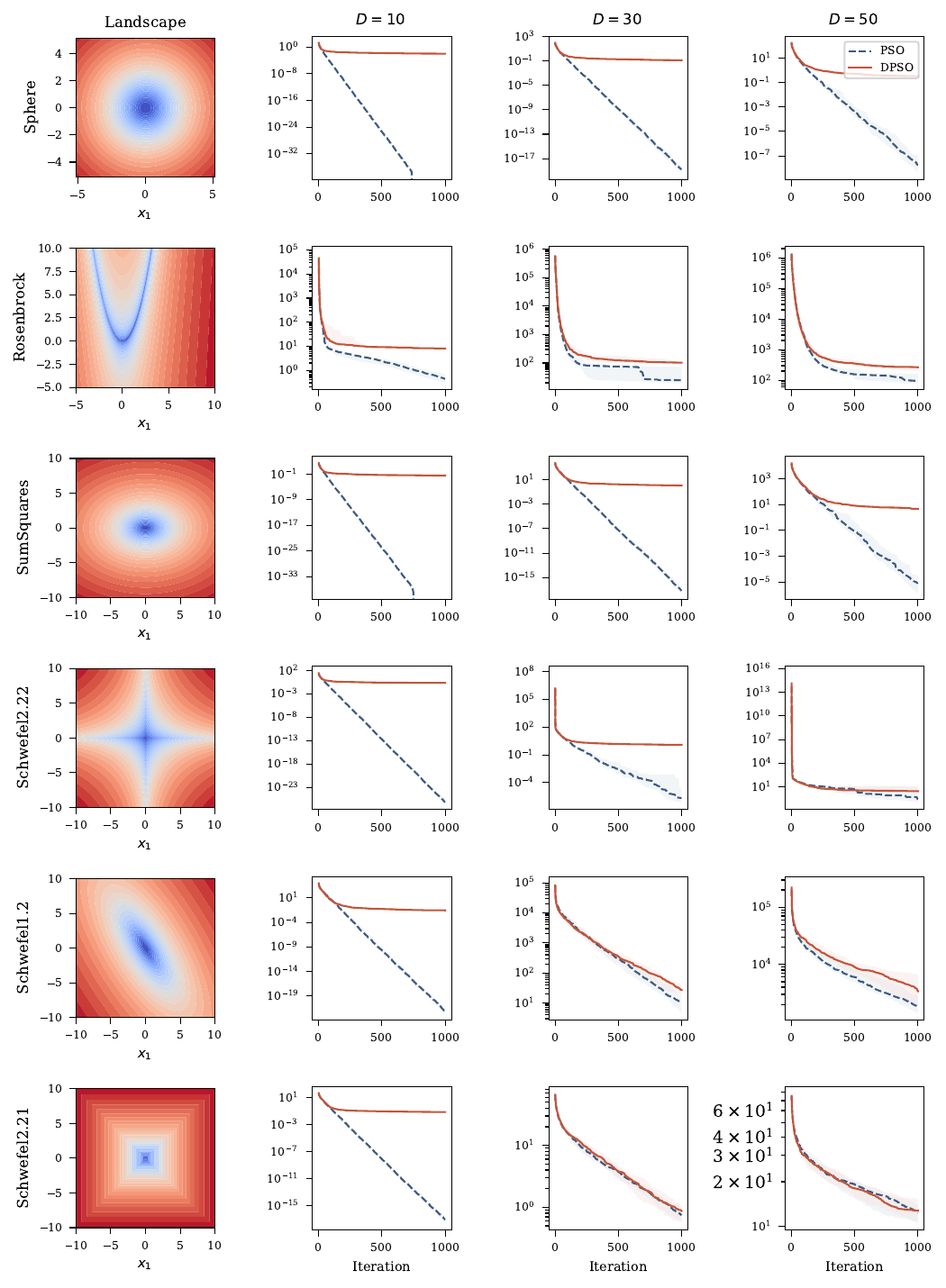}
    \caption{Unimodal benchmarks (part~1).
    Left column: 2D function landscape (log-scaled contour).
    Remaining columns: convergence curves for {\color{cRPNI}PSO} (dashed) and
    {\color{cOurs}DPSO} (solid) at $D \in \{10, 30, 50\}$.
    Lines show the median over 30 runs; shaded bands show the IQR.}
    \label{fig:convergence_unimodal}
\end{figure}

\begin{figure}[H]
    \centering
    \includegraphics[width=\textwidth]{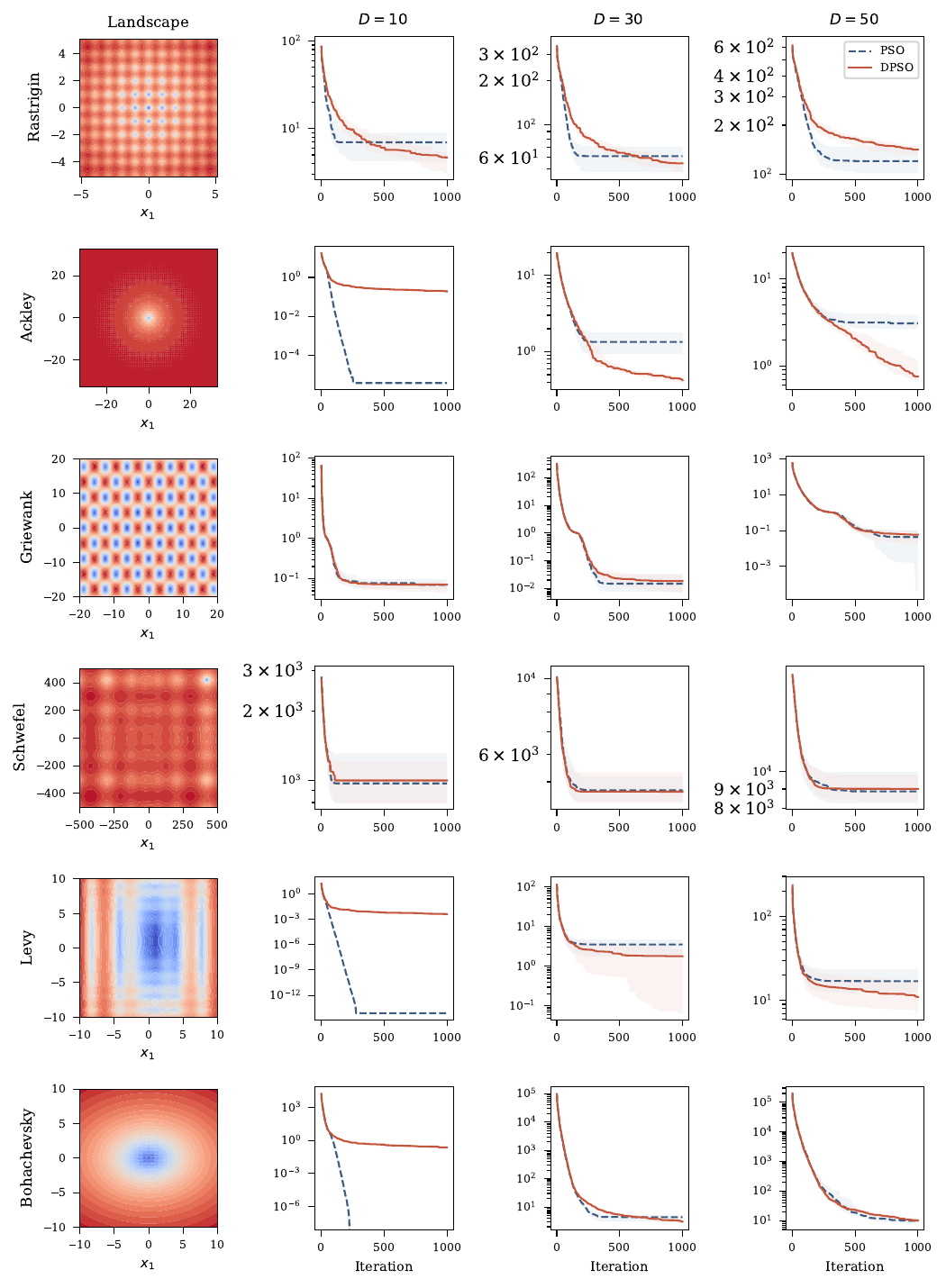}
    \caption{Multimodal benchmarks (part~1). Layout as in Fig.~\ref{fig:convergence_unimodal}.}
    \label{fig:convergence_multimodal}
\end{figure}

\begin{figure}[H]
    \centering
    \includegraphics[width=\textwidth]{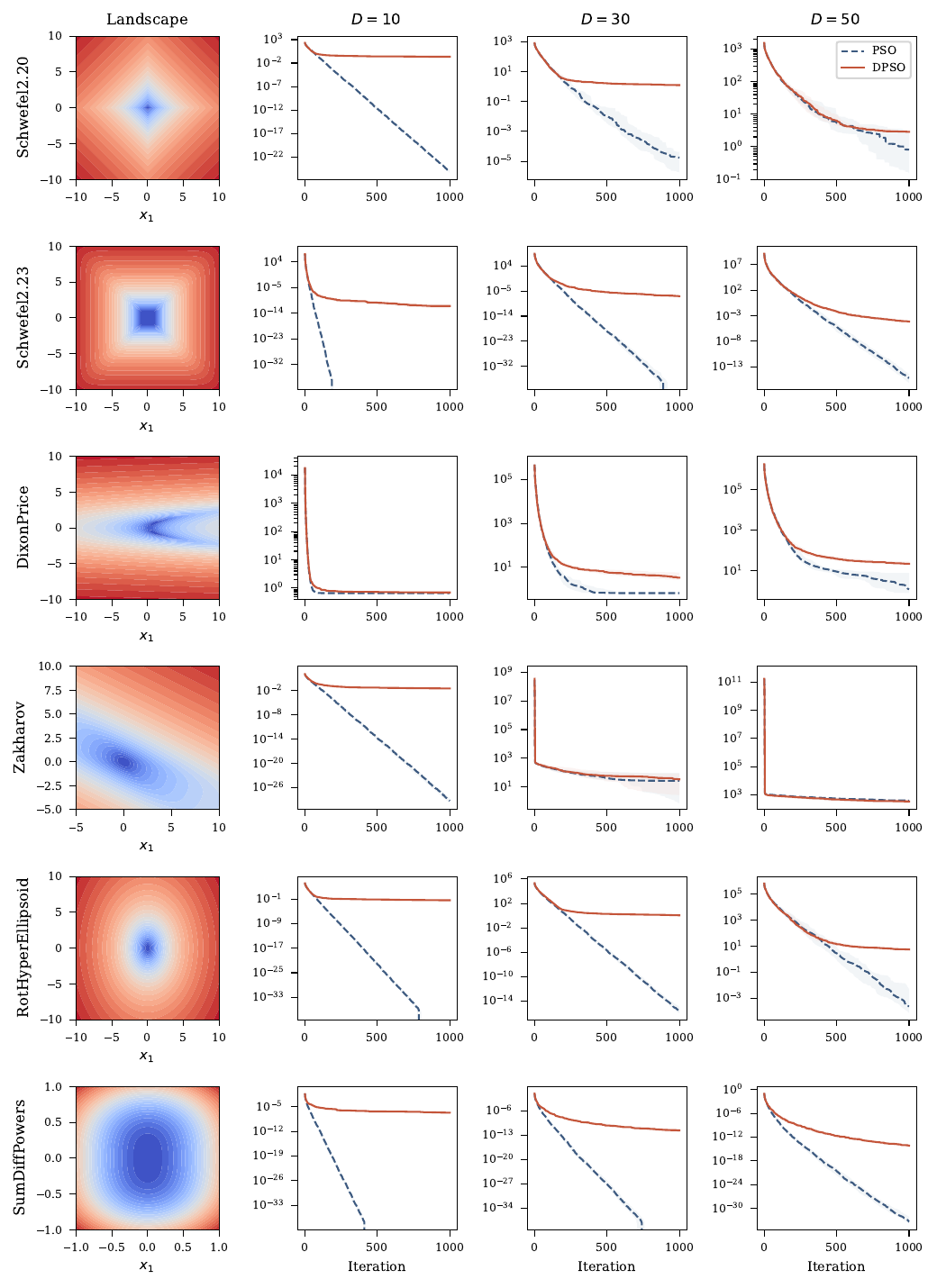}
    \caption{Unimodal benchmarks (part~2).}
    \label{fig:convergence_unimodal_2}
\end{figure}

\begin{figure}[H]
    \centering
    \includegraphics[width=\textwidth]{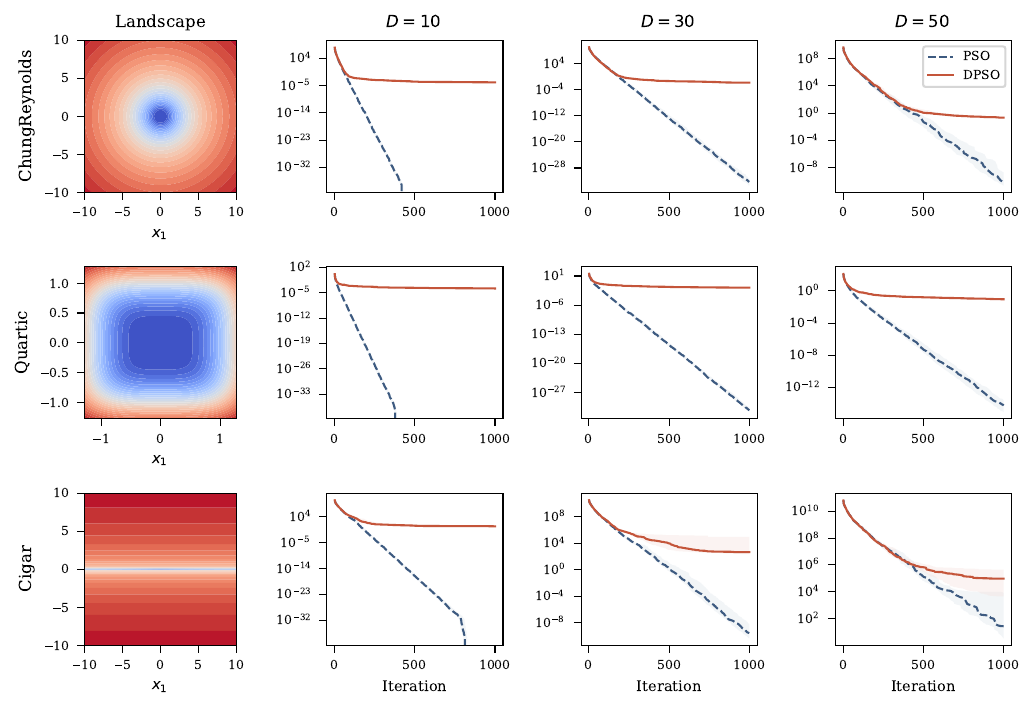}
    \caption{Unimodal benchmarks (part~3).}
    \label{fig:convergence_unimodal_3}
\end{figure}

\begin{figure}[H]
    \centering
    \includegraphics[width=\textwidth]{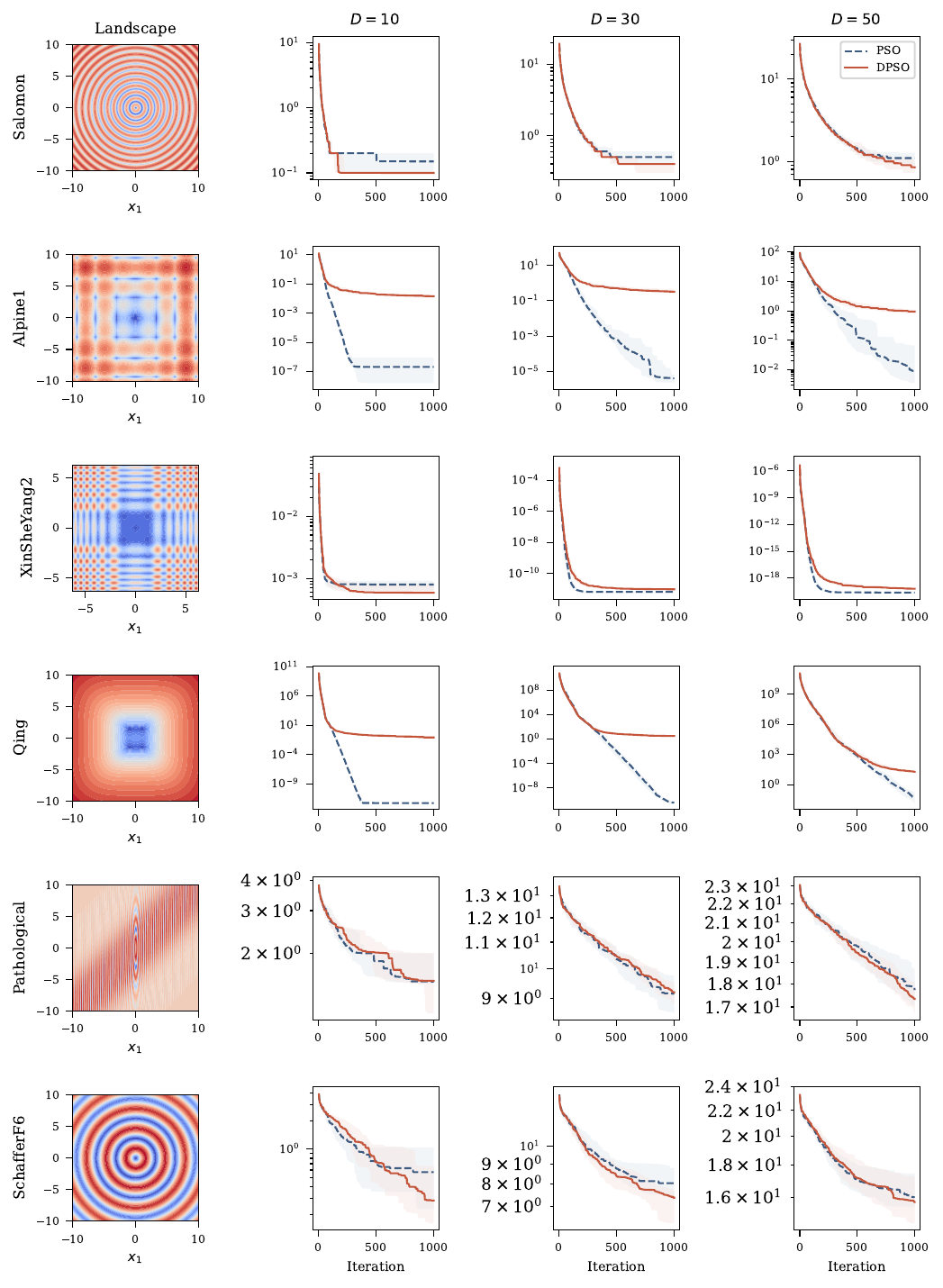}
    \caption{Multimodal benchmarks (part~2).}
    \label{fig:convergence_multimodal_2}
\end{figure}

\begin{figure}[H]
    \centering
    \includegraphics[width=\textwidth]{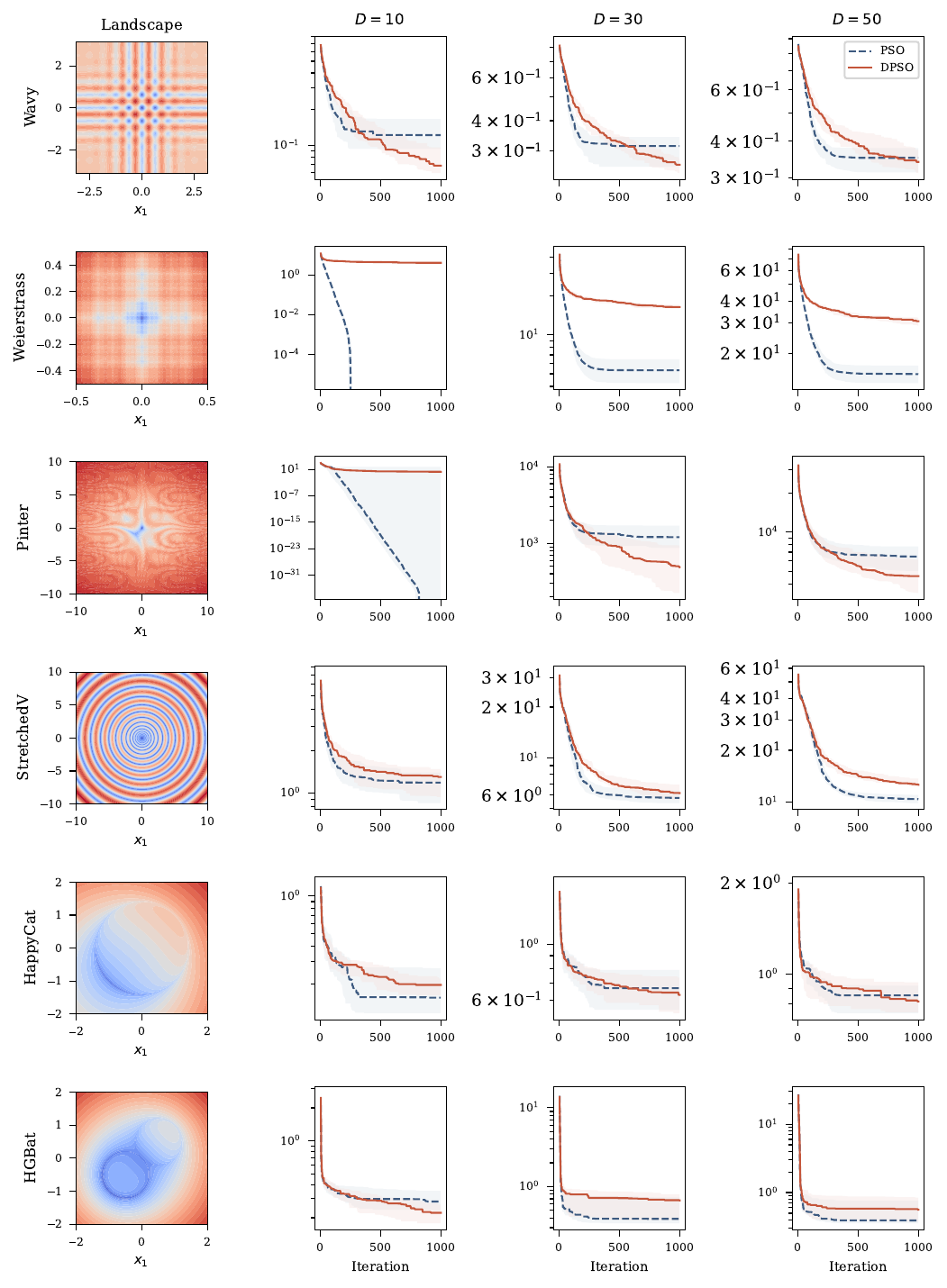}
    \caption{Multimodal benchmarks (part~3).}
    \label{fig:convergence_multimodal_3}
\end{figure}

\begin{figure}[H]
    \centering
    \includegraphics[width=\textwidth]{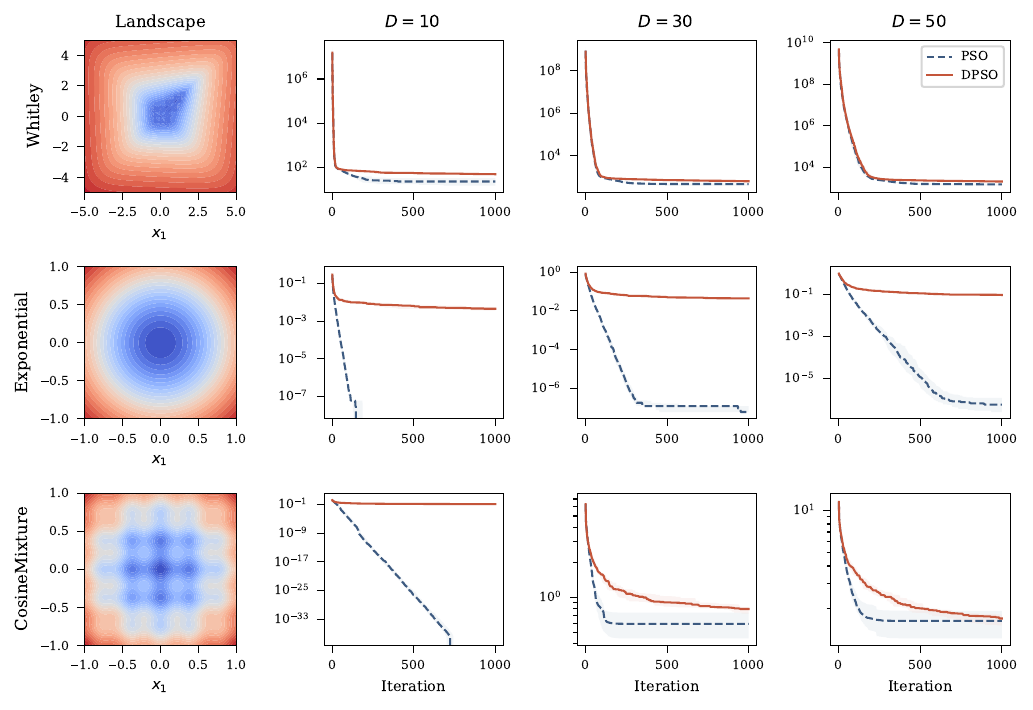}
    \caption{Multimodal benchmarks (part~4).}
    \label{fig:convergence_multimodal_4}
\end{figure}

\section{Benchmark functions}
\label{app:functions}

The Table \ref{tab:benchmarks} presents a comprehensive summary of the benchmark functions used

\begin{table}[H]
\centering
\renewcommand{\arraystretch}{1.3}
\caption{Benchmark functions used in the experiments ($D$-dimensional).
All global minima satisfy $f^* = 0$ except Schwefel, where
$f^* \approx 0$.}
\label{tab:benchmarks}
\scriptsize
\setlength{\tabcolsep}{4pt}
\begin{tabularx}{\textwidth}{l >{\raggedright\arraybackslash}X c c}
\toprule
Function & Definition & Bounds & $f^*$ \\
\midrule
\rowcolor{cGroupHd}
\multicolumn{4}{l}{\textit{Unimodal}} \\
Sphere            & $\sum x_i^2$ & $[-5.12,\;5.12]^D$     & $0$ \\
Rosenbrock        & $\sum\bigl[100(x_{i+1}\!-\!x_i^2)^2+(1\!-\!x_i)^2\bigr]$ & $[-5,\;10]^D$ & $0$ \\
SumSquares        & $\sum i\,x_i^2$ & $[-10,\;10]^D$       & $0$ \\
Schwefel\,2.22    & $\sum|x_i|+\prod|x_i|$ & $[-10,\;10]^D$ & $0$ \\
Schwefel\,1.2     & $\sum_{i}\bigl(\sum_{j=1}^{i}x_j\bigr)^{\!2}$ & $[-100,\;100]^D$ & $0$ \\
Schwefel\,2.21    & $\max_i|x_i|$ & $[-100,\;100]^D$      & $0$ \\
Schwefel\,2.20    & $\sum|x_i|$ & $[-100,\;100]^D$        & $0$ \\
Schwefel\,2.23    & $\sum x_i^{10}$ & $[-10,\;10]^D$      & $0$ \\
DixonPrice        & $(x_1\!-\!1)^2+\sum_{i=2}^{D}i(2x_i^2\!-\!x_{i-1})^2$ & $[-10,\;10]^D$ & $0$ \\
Zakharov          & $\sum x_i^2+\bigl(\sum\tfrac{i}{2}x_i\bigr)^{\!2}+\bigl(\sum\tfrac{i}{2}x_i\bigr)^{\!4}$ & $[-5,\;10]^D$ & $0$ \\
RotHyperEllipsoid & $\sum(D\!-\!i\!+\!1)\,x_i^2$ & $[-65.536,\;65.536]^D$ & $0$ \\
SumDiffPowers     & $\sum|x_i|^{i+1}$ & $[-1,\;1]^D$      & $0$ \\
ChungReynolds     & $\bigl(\sum x_i^2\bigr)^{\!2}$ & $[-100,\;100]^D$ & $0$ \\
Quartic           & $\sum i\,x_i^4$ & $[-1.28,\;1.28]^D$  & $0$ \\
Cigar             & $x_1^2+10^6\!\sum_{i=2}^{D}x_i^2$ & $[-100,\;100]^D$ & $0$ \\
\midrule
\rowcolor{cGroupHd}
\multicolumn{4}{l}{\textit{Multimodal}} \\
Rastrigin     & $10D+\sum\bigl[x_i^2-10\cos(2\pi x_i)\bigr]$ & $[-5.12,\;5.12]^D$ & $0$ \\
Ackley        & $-20e^{-0.2\sqrt{\frac{1}{D}\sum x_i^2}}\!-e^{\frac{1}{D}\sum\cos 2\pi x_i}\!+\!20\!+\!e$ & $[-32.768,\;32.768]^D$ & $0$ \\
Griewank      & $1+\tfrac{1}{4000}\sum x_i^2-\prod\cos\!\bigl(\tfrac{x_i}{\sqrt{i}}\bigr)$ & $[-600,\;600]^D$ & $0$ \\
Schwefel      & $418.9829D-\sum x_i\sin\!\bigl(\!\sqrt{|x_i|}\bigr)$ & $[-500,\;500]^D$ & ${\approx}\,0$ \\
Levy          & $\sin^2\!(\pi w_1)\!+\!\sum(w_i\!\!-\!\!1)^2[1\!+\!10\sin^2\!(\pi w_i\!\!+\!\!1)]+(w_D\!\!-\!\!1)^2[1\!+\!\sin^2\!(2\pi w_D)]$;\newline $w_i\!=\!1\!+\!\tfrac{x_i-1}{4}$ & $[-10,\;10]^D$ & $0$ \\
Bohachevsky   & $\sum\bigl[x_i^2\!+\!2x_{i\!+\!1}^2\!-\!0.3\cos 3\pi x_i\!-\!0.4\cos 4\pi x_{i\!+\!1}\!+\!0.7\bigr]$ & $[-100,\;100]^D$ & $0$ \\
Salomon       & $1-\cos\!\bigl(2\pi\!\sqrt{\sum x_i^2}\bigr)+0.1\sqrt{\sum x_i^2}$ & $[-100,\;100]^D$ & $0$ \\
Alpine1       & $\sum|x_i\sin(x_i)+0.1\,x_i|$ & $[-10,\;10]^D$ & $0$ \\
XinSheYang2   & $\bigl(\sum|x_i|\bigr)\exp\!\bigl(\!-\!\sum\sin x_i^2\bigr)$ & $[-2\pi,\;2\pi]^D$ & $0$ \\
Qing          & $\sum(x_i^2-i)^2$ & $[-500,\;500]^D$    & $0$ \\
Pathological  & $\sum\!\bigl[0.5+\tfrac{\sin^2\!\sqrt{100x_i^2+x_{i\!+\!1}^2}\,-\,0.5}{1+0.001(x_i-x_{i\!+\!1})^4}\bigr]$ & $[-100,\;100]^D$ & $0$ \\
SchafferF6    & $\sum\!\bigl[0.5+\tfrac{\sin^2\!\sqrt{x_i^2+x_{i\!+\!1}^2}\,-\,0.5}{(1+0.001(x_i^2+x_{i\!+\!1}^2))^2}\bigr]$ & $[-100,\;100]^D$ & $0$ \\
Wavy          & $1-\tfrac{1}{D}\sum\cos(10x_i)\exp(-x_i^2/2)$ & $[-\pi,\;\pi]^D$ & $0$ \\
Weierstrass   & $\sum_i\!\bigl[\sum_{k=0}^{20}a^k\!\cos\bigl(2\pi b^k(x_i\!+\!0.5)\bigr)\bigr]\!-\!D\!\sum_{k=0}^{20}a^k\!\cos(\pi b^k)$;\newline $a\!=\!0.5,\; b\!=\!3$ & $[-0.5,\;0.5]^D$ & $0$ \\
Pinter        & $\sum ix_i^2+20\!\sum i\sin^2\!\!A_i+\sum i\log_{10}(1\!+\!iB_i^2)$;\newline $A_i\!=\!x_{i\!-\!1}\!\sin x_i\!+\!\sin x_{i\!+\!1}$,\; $B_i\!=\!x_{i\!-\!1}^2\!\!-\!2x_i\!+\!3x_{i\!+\!1}\!\!-\!\cos x_i\!+\!1$ & $[-10,\;10]^D$ & $0$ \\
StretchedV    & $\sum t_i^{1/4}\bigl[\sin^2(50\,t_i^{1/10})\!+\!0.1\bigr]$;\; $t_i\!=\!x_i^2\!+\!x_{i+1}^2$ & $[-10,\;10]^D$ & $0$ \\
HappyCat      & $|\!\sum x_i^2\!-\!D|^{1/4}\!+\!(0.5\!\sum x_i^2\!+\!\sum x_i)/D+0.5$ & $[-2,\;2]^D$ & $0$ \\
HGBat         & $|(\!\sum x_i^2)^2\!-\!(\!\sum x_i)^2|^{1/2}\!+\!(0.5\!\sum x_i^2\!+\!\sum x_i)/D+0.5$ & $[-2,\;2]^D$ & $0$ \\
Whitley       & $\sum_i\!\sum_j\bigl[y_{ij}^2/4000-\cos y_{ij}+1\bigr]$;\newline $y_{ij}\!=\!100(x_i^2\!-\!x_j)^2\!+\!(1\!-\!x_j)^2$ & $[-10.24,\;10.24]^D$ & $0$ \\
Exponential   & $1-\exp(-0.5\sum x_i^2)$ & $[-1,\;1]^D$   & $0$ \\
CosineMixture & $\sum\bigl[x_i^2+0.1(1\!-\!\cos 5\pi x_i)\bigr]$ & $[-1,\;1]^D$ & $0$ \\
\bottomrule
\end{tabularx}
\end{table}

\end{document}

%% file: outputs/tables/results_unimodal.tex
\begin{table}[t]
\centering
\caption{Best fitness values on \textbf{unimodal} benchmarks (mean $\pm$ standard deviation over 30 runs). \colorbox{cBest}{\textbf{Shaded}} indicates the better result.}
\label{tab:results-unimodal}
\footnotesize
\setlength{\tabcolsep}{3.5pt}
\resizebox{\textwidth}{!}{%
\begin{tabular}{l cc cc cc}
\toprule
\rowcolor{cGroupHd}
Function & \multicolumn{2}{c}{$D=10$} & \multicolumn{2}{c}{$D=30$} & \multicolumn{2}{c}{$D=50$} \\
\cmidrule(lr){2-3}\cmidrule(lr){4-5}\cmidrule(lr){6-7}
\rowcolor{cGroupHd}
 & PSO & DPSO & PSO & DPSO & PSO & DPSO \\
\midrule
Sphere & \cellcolor{cBest}\textbf{$0 \pm 0$} & $1.10\!\times\!10^{-2} \pm 3.60\!\times\!10^{-3}$ & \cellcolor{cBest}\textbf{$2.85\!\times\!10^{-19} \pm 3.73\!\times\!10^{-19}$} & $1.30\!\times\!10^{-1} \pm 2.34\!\times\!10^{-2}$ & \cellcolor{cBest}\textbf{$1.99\!\times\!10^{-6} \pm 6.72\!\times\!10^{-6}$} & $3.31\!\times\!10^{-1} \pm 4.96\!\times\!10^{-2}$ \\
Rosenbrock & \cellcolor{cBest}\textbf{$6.36 \pm 2.90\!\times\!10^{1}$} & $7.83 \pm 1.19$ & \cellcolor{cBest}\textbf{$4.23\!\times\!10^{1} \pm 3.02\!\times\!10^{1}$} & $2.75\!\times\!10^{2} \pm 6.16\!\times\!10^{2}$ & $1.28\!\times\!10^{4} \pm 2.31\!\times\!10^{4}$ & \cellcolor{cBest}\textbf{$4.82\!\times\!10^{3} \pm 1.43\!\times\!10^{4}$} \\
SumSquares & \cellcolor{cBest}\textbf{$0 \pm 0$} & $3.75\!\times\!10^{-2} \pm 1.27\!\times\!10^{-2}$ & \cellcolor{cBest}\textbf{$5.88\!\times\!10^{-17} \pm 1.44\!\times\!10^{-16}$} & $4.43 \pm 1.80\!\times\!10^{1}$ & $5.33\!\times\!10^{1} \pm 1.31\!\times\!10^{2}$ & \cellcolor{cBest}\textbf{$4.48\!\times\!10^{1} \pm 1.17\!\times\!10^{2}$} \\
Schwefel2.22 & \cellcolor{cBest}\textbf{$1.47\!\times\!10^{-26} \pm 1.76\!\times\!10^{-26}$} & $2.30\!\times\!10^{-1} \pm 2.80\!\times\!10^{-2}$ & \cellcolor{cBest}\textbf{$1.00 \pm 3.00$} & $1.86 \pm 2.48$ & \cellcolor{cBest}\textbf{$4.47 \pm 5.57$} & $6.53 \pm 5.50$ \\
Schwefel1.2 & \cellcolor{cBest}\textbf{$9.57\!\times\!10^{-22} \pm 3.07\!\times\!10^{-21}$} & $2.59\!\times\!10^{-2} \pm 1.01\!\times\!10^{-2}$ & $5.10\!\times\!10^{2} \pm 1.50\!\times\!10^{3}$ & \cellcolor{cBest}\textbf{$1.96\!\times\!10^{2} \pm 8.95\!\times\!10^{2}$} & \cellcolor{cBest}\textbf{$2.80\!\times\!10^{3} \pm 2.24\!\times\!10^{3}$} & $5.38\!\times\!10^{3} \pm 4.74\!\times\!10^{3}$ \\
Schwefel2.21 & \cellcolor{cBest}\textbf{$1.40\!\times\!10^{-17} \pm 1.80\!\times\!10^{-17}$} & $7.28\!\times\!10^{-2} \pm 1.34\!\times\!10^{-2}$ & \cellcolor{cBest}\textbf{$8.33\!\times\!10^{-1} \pm 3.94\!\times\!10^{-1}$} & $9.57\!\times\!10^{-1} \pm 5.12\!\times\!10^{-1}$ & \cellcolor{cBest}\textbf{$1.24\!\times\!10^{1} \pm 2.65$} & $1.28\!\times\!10^{1} \pm 2.87$ \\
Schwefel2.20 & \cellcolor{cBest}\textbf{$1.45\!\times\!10^{-25} \pm 1.98\!\times\!10^{-25}$} & $2.14\!\times\!10^{-1} \pm 4.23\!\times\!10^{-2}$ & \cellcolor{cBest}\textbf{$9.49\!\times\!10^{-4} \pm 3.36\!\times\!10^{-3}$} & $4.57 \pm 1.79\!\times\!10^{1}$ & $5.42 \pm 1.80\!\times\!10^{1}$ & \cellcolor{cBest}\textbf{$2.99 \pm 5.79\!\times\!10^{-1}$} \\
Schwefel2.23 & \cellcolor{cBest}\textbf{$0 \pm 0$} & $8.86\!\times\!10^{-12} \pm 1.58\!\times\!10^{-11}$ & \cellcolor{cBest}\textbf{$0 \pm 0$} & $3.48\!\times\!10^{-7} \pm 4.68\!\times\!10^{-7}$ & \cellcolor{cBest}\textbf{$8.57\!\times\!10^{-14} \pm 3.38\!\times\!10^{-13}$} & $9.67\!\times\!10^{-5} \pm 1.00\!\times\!10^{-4}$ \\
DixonPrice & \cellcolor{cBest}\textbf{$6.00\!\times\!10^{-1} \pm 2.00\!\times\!10^{-1}$} & $6.97\!\times\!10^{-1} \pm 3.85\!\times\!10^{-2}$ & \cellcolor{cBest}\textbf{$6.56 \pm 2.21\!\times\!10^{1}$} & $1.01\!\times\!10^{1} \pm 2.29\!\times\!10^{1}$ & \cellcolor{cBest}\textbf{$4.04\!\times\!10^{1} \pm 9.65\!\times\!10^{1}$} & $4.61\!\times\!10^{1} \pm 8.15\!\times\!10^{1}$ \\
Zakharov & $1.12 \pm 6.04$ & \cellcolor{cBest}\textbf{$3.49\!\times\!10^{-2} \pm 1.11\!\times\!10^{-2}$} & \cellcolor{cBest}\textbf{$5.67\!\times\!10^{1} \pm 7.86\!\times\!10^{1}$} & $6.09\!\times\!10^{1} \pm 6.84\!\times\!10^{1}$ & $3.70\!\times\!10^{2} \pm 1.62\!\times\!10^{2}$ & \cellcolor{cBest}\textbf{$3.49\!\times\!10^{2} \pm 1.58\!\times\!10^{2}$} \\
RotHyperEllipsoid & \cellcolor{cBest}\textbf{$0 \pm 0$} & $4.04\!\times\!10^{-2} \pm 1.27\!\times\!10^{-2}$ & \cellcolor{cBest}\textbf{$1.90\!\times\!10^{-15} \pm 3.86\!\times\!10^{-15}$} & $4.31\!\times\!10^{2} \pm 1.70\!\times\!10^{3}$ & $2.00\!\times\!10^{3} \pm 6.54\!\times\!10^{3}$ & \cellcolor{cBest}\textbf{$5.78\!\times\!10^{2} \pm 1.83\!\times\!10^{3}$} \\
SumDiffPowers & \cellcolor{cBest}\textbf{$0 \pm 0$} & $2.55\!\times\!10^{-7} \pm 1.86\!\times\!10^{-7}$ & \cellcolor{cBest}\textbf{$0 \pm 0$} & $5.31\!\times\!10^{-12} \pm 6.91\!\times\!10^{-12}$ & \cellcolor{cBest}\textbf{$4.30\!\times\!10^{-31} \pm 2.25\!\times\!10^{-30}$} & $1.64\!\times\!10^{-14} \pm 2.68\!\times\!10^{-14}$ \\
ChungReynolds & \cellcolor{cBest}\textbf{$0 \pm 0$} & $1.30\!\times\!10^{-4} \pm 7.64\!\times\!10^{-5}$ & \cellcolor{cBest}\textbf{$2.48\!\times\!10^{-28} \pm 1.15\!\times\!10^{-27}$} & $1.71\!\times\!10^{-2} \pm 6.07\!\times\!10^{-3}$ & \cellcolor{cBest}\textbf{$6.51\!\times\!10^{-7} \pm 3.27\!\times\!10^{-6}$} & $2.77\!\times\!10^{-1} \pm 3.46\!\times\!10^{-1}$ \\
Quartic & \cellcolor{cBest}\textbf{$0 \pm 0$} & $1.57\!\times\!10^{-4} \pm 9.09\!\times\!10^{-5}$ & \cellcolor{cBest}\textbf{$1.28\!\times\!10^{-30} \pm 2.96\!\times\!10^{-30}$} & $1.59\!\times\!10^{-2} \pm 4.79\!\times\!10^{-3}$ & $6.26\!\times\!10^{-1} \pm 2.56$ & \cellcolor{cBest}\textbf{$9.65\!\times\!10^{-2} \pm 2.26\!\times\!10^{-2}$} \\
Cigar & \cellcolor{cBest}\textbf{$0 \pm 0$} & $1.71\!\times\!10^{3} \pm 5.21\!\times\!10^{3}$ & \cellcolor{cBest}\textbf{$1.33\!\times\!10^{3} \pm 3.40\!\times\!10^{3}$} & $4.06\!\times\!10^{4} \pm 6.40\!\times\!10^{4}$ & \cellcolor{cBest}\textbf{$4.69\!\times\!10^{3} \pm 1.22\!\times\!10^{4}$} & $2.13\!\times\!10^{5} \pm 2.26\!\times\!10^{5}$ \\
\bottomrule
\end{tabular}}
\end{table}

%% file: outputs/tables/results_multimodal.tex
\begin{table}[t]
\centering
\caption{Best fitness values on \textbf{multimodal} benchmarks (mean $\pm$ standard deviation over 30 runs). \colorbox{cBest}{\textbf{Shaded}} indicates the better result.}
\label{tab:results-multimodal}
\footnotesize
\setlength{\tabcolsep}{3.5pt}
\resizebox{\textwidth}{!}{%
\begin{tabular}{l cc cc cc}
\toprule
\rowcolor{cGroupHd}
Function & \multicolumn{2}{c}{$D=10$} & \multicolumn{2}{c}{$D=30$} & \multicolumn{2}{c}{$D=50$} \\
\cmidrule(lr){2-3}\cmidrule(lr){4-5}\cmidrule(lr){6-7}
\rowcolor{cGroupHd}
 & PSO & DPSO & PSO & DPSO & PSO & DPSO \\
\midrule
Rastrigin & $7.00 \pm 3.14$ & \cellcolor{cBest}\textbf{$4.36 \pm 1.58$} & $5.98\!\times\!10^{1} \pm 1.52\!\times\!10^{1}$ & \cellcolor{cBest}\textbf{$5.36\!\times\!10^{1} \pm 1.00\!\times\!10^{1}$} & \cellcolor{cBest}\textbf{$1.27\!\times\!10^{2} \pm 2.72\!\times\!10^{1}$} & $1.39\!\times\!10^{2} \pm 1.94\!\times\!10^{1}$ \\
Ackley & \cellcolor{cBest}\textbf{$3.05\!\times\!10^{-6} \pm 1.53\!\times\!10^{-6}$} & $1.93\!\times\!10^{-1} \pm 4.06\!\times\!10^{-2}$ & $1.20 \pm 7.81\!\times\!10^{-1}$ & \cellcolor{cBest}\textbf{$4.34\!\times\!10^{-1} \pm 5.78\!\times\!10^{-2}$} & $3.27 \pm 9.33\!\times\!10^{-1}$ & \cellcolor{cBest}\textbf{$8.98\!\times\!10^{-1} \pm 3.62\!\times\!10^{-1}$} \\
Griewank & $7.97\!\times\!10^{-2} \pm 4.11\!\times\!10^{-2}$ & \cellcolor{cBest}\textbf{$6.87\!\times\!10^{-2} \pm 3.20\!\times\!10^{-2}$} & \cellcolor{cBest}\textbf{$2.23\!\times\!10^{-2} \pm 2.21\!\times\!10^{-2}$} & $2.84\!\times\!10^{-2} \pm 3.01\!\times\!10^{-2}$ & $1.19\!\times\!10^{-1} \pm 3.23\!\times\!10^{-1}$ & \cellcolor{cBest}\textbf{$7.02\!\times\!10^{-2} \pm 6.74\!\times\!10^{-2}$} \\
Schwefel & $1.03\!\times\!10^{3} \pm 3.28\!\times\!10^{2}$ & \cellcolor{cBest}\textbf{$9.86\!\times\!10^{2} \pm 2.57\!\times\!10^{2}$} & \cellcolor{cBest}\textbf{$4.76\!\times\!10^{3} \pm 7.59\!\times\!10^{2}$} & $4.79\!\times\!10^{3} \pm 6.41\!\times\!10^{2}$ & $8.99\!\times\!10^{3} \pm 1.04\!\times\!10^{3}$ & \cellcolor{cBest}\textbf{$8.98\!\times\!10^{3} \pm 1.08\!\times\!10^{3}$} \\
Levy & \cellcolor{cBest}\textbf{$7.64\!\times\!10^{-15} \pm 0$} & $4.05\!\times\!10^{-3} \pm 1.29\!\times\!10^{-3}$ & $4.18 \pm 3.56$ & \cellcolor{cBest}\textbf{$1.60 \pm 1.69$} & $1.73\!\times\!10^{1} \pm 6.81$ & \cellcolor{cBest}\textbf{$1.19\!\times\!10^{1} \pm 5.97$} \\
Bohachevsky & \cellcolor{cBest}\textbf{$4.88\!\times\!10^{-2} \pm 2.00\!\times\!10^{-1}$} & $2.58\!\times\!10^{-1} \pm 9.58\!\times\!10^{-2}$ & $4.52 \pm 2.02$ & \cellcolor{cBest}\textbf{$3.11 \pm 6.03\!\times\!10^{-1}$} & $1.01\!\times\!10^{1} \pm 2.94$ & \cellcolor{cBest}\textbf{$1.01\!\times\!10^{1} \pm 2.10$} \\
Salomon & $1.53\!\times\!10^{-1} \pm 5.62\!\times\!10^{-2}$ & \cellcolor{cBest}\textbf{$1.10\!\times\!10^{-1} \pm 3.00\!\times\!10^{-2}$} & $5.07\!\times\!10^{-1} \pm 1.46\!\times\!10^{-1}$ & \cellcolor{cBest}\textbf{$3.70\!\times\!10^{-1} \pm 7.37\!\times\!10^{-2}$} & $1.16 \pm 4.12\!\times\!10^{-1}$ & \cellcolor{cBest}\textbf{$9.33\!\times\!10^{-1} \pm 2.71\!\times\!10^{-1}$} \\
Alpine1 & \cellcolor{cBest}\textbf{$4.92\!\times\!10^{-7} \pm 5.48\!\times\!10^{-7}$} & $1.55\!\times\!10^{-2} \pm 5.74\!\times\!10^{-3}$ & \cellcolor{cBest}\textbf{$2.96\!\times\!10^{-1} \pm 1.11$} & $5.04\!\times\!10^{-1} \pm 8.18\!\times\!10^{-1}$ & \cellcolor{cBest}\textbf{$4.77\!\times\!10^{-1} \pm 1.38$} & $1.24 \pm 1.13$ \\
XinSheYang2 & $8.58\!\times\!10^{-4} \pm 3.47\!\times\!10^{-4}$ & \cellcolor{cBest}\textbf{$7.13\!\times\!10^{-4} \pm 2.01\!\times\!10^{-4}$} & \cellcolor{cBest}\textbf{$6.79\!\times\!10^{-12} \pm 1.54\!\times\!10^{-12}$} & $9.43\!\times\!10^{-12} \pm 1.43\!\times\!10^{-12}$ & \cellcolor{cBest}\textbf{$2.53\!\times\!10^{-20} \pm 5.45\!\times\!10^{-21}$} & $6.96\!\times\!10^{-20} \pm 2.05\!\times\!10^{-20}$ \\
Qing & \cellcolor{cBest}\textbf{$5.12\!\times\!10^{-13} \pm 0$} & $1.02\!\times\!10^{-1} \pm 4.49\!\times\!10^{-2}$ & \cellcolor{cBest}\textbf{$5.19\!\times\!10^{-11} \pm 3.97\!\times\!10^{-11}$} & $3.11 \pm 9.89\!\times\!10^{-1}$ & \cellcolor{cBest}\textbf{$1.21\!\times\!10^{-1} \pm 2.09\!\times\!10^{-1}$} & $2.35\!\times\!10^{1} \pm 1.20\!\times\!10^{1}$ \\
Pathological & $1.63 \pm 4.81\!\times\!10^{-1}$ & \cellcolor{cBest}\textbf{$1.63 \pm 4.99\!\times\!10^{-1}$} & \cellcolor{cBest}\textbf{$9.01 \pm 8.35\!\times\!10^{-1}$} & $9.02 \pm 7.82\!\times\!10^{-1}$ & $1.77\!\times\!10^{1} \pm 1.66$ & \cellcolor{cBest}\textbf{$1.74\!\times\!10^{1} \pm 9.16\!\times\!10^{-1}$} \\
SchafferF6 & $6.39\!\times\!10^{-1} \pm 4.33\!\times\!10^{-1}$ & \cellcolor{cBest}\textbf{$4.64\!\times\!10^{-1} \pm 3.99\!\times\!10^{-1}$} & $7.94 \pm 1.12$ & \cellcolor{cBest}\textbf{$7.05 \pm 1.06$} & $1.62\!\times\!10^{1} \pm 1.35$ & \cellcolor{cBest}\textbf{$1.57\!\times\!10^{1} \pm 1.72$} \\
Wavy & $1.29\!\times\!10^{-1} \pm 5.93\!\times\!10^{-2}$ & \cellcolor{cBest}\textbf{$8.54\!\times\!10^{-2} \pm 4.52\!\times\!10^{-2}$} & $3.09\!\times\!10^{-1} \pm 5.87\!\times\!10^{-2}$ & \cellcolor{cBest}\textbf{$2.66\!\times\!10^{-1} \pm 3.94\!\times\!10^{-2}$} & \cellcolor{cBest}\textbf{$3.48\!\times\!10^{-1} \pm 5.76\!\times\!10^{-2}$} & $3.48\!\times\!10^{-1} \pm 4.84\!\times\!10^{-2}$ \\
Weierstrass & \cellcolor{cBest}\textbf{$7.86\!\times\!10^{-2} \pm 3.29\!\times\!10^{-1}$} & $4.00 \pm 3.97\!\times\!10^{-1}$ & \cellcolor{cBest}\textbf{$5.16 \pm 1.60$} & $1.65\!\times\!10^{1} \pm 1.01$ & \cellcolor{cBest}\textbf{$1.53\!\times\!10^{1} \pm 3.53$} & $3.05\!\times\!10^{1} \pm 1.83$ \\
Pinter & $3.25\!\times\!10^{1} \pm 4.21\!\times\!10^{1}$ & \cellcolor{cBest}\textbf{$3.88 \pm 1.05\!\times\!10^{1}$} & $1.37\!\times\!10^{3} \pm 6.75\!\times\!10^{2}$ & \cellcolor{cBest}\textbf{$6.18\!\times\!10^{2} \pm 4.73\!\times\!10^{2}$} & $6.52\!\times\!10^{3} \pm 1.72\!\times\!10^{3}$ & \cellcolor{cBest}\textbf{$4.93\!\times\!10^{3} \pm 1.64\!\times\!10^{3}$} \\
StretchedV & \cellcolor{cBest}\textbf{$1.16 \pm 4.11\!\times\!10^{-1}$} & $1.25 \pm 3.60\!\times\!10^{-1}$ & \cellcolor{cBest}\textbf{$5.82 \pm 9.14\!\times\!10^{-1}$} & $6.31 \pm 8.72\!\times\!10^{-1}$ & \cellcolor{cBest}\textbf{$1.06\!\times\!10^{1} \pm 9.11\!\times\!10^{-1}$} & $1.28\!\times\!10^{1} \pm 1.22$ \\
HappyCat & $1.99\!\times\!10^{-1} \pm 1.12\!\times\!10^{-1}$ & \cellcolor{cBest}\textbf{$1.92\!\times\!10^{-1} \pm 4.18\!\times\!10^{-2}$} & $6.62\!\times\!10^{-1} \pm 1.84\!\times\!10^{-1}$ & \cellcolor{cBest}\textbf{$6.35\!\times\!10^{-1} \pm 1.58\!\times\!10^{-1}$} & $8.46\!\times\!10^{-1} \pm 1.25\!\times\!10^{-1}$ & \cellcolor{cBest}\textbf{$8.35\!\times\!10^{-1} \pm 1.13\!\times\!10^{-1}$} \\
HGBat & $2.86\!\times\!10^{-1} \pm 7.00\!\times\!10^{-2}$ & \cellcolor{cBest}\textbf{$2.24\!\times\!10^{-1} \pm 6.26\!\times\!10^{-2}$} & \cellcolor{cBest}\textbf{$5.38\!\times\!10^{-1} \pm 2.50\!\times\!10^{-1}$} & $5.92\!\times\!10^{-1} \pm 2.28\!\times\!10^{-1}$ & \cellcolor{cBest}\textbf{$5.44\!\times\!10^{-1} \pm 2.34\!\times\!10^{-1}$} & $6.12\!\times\!10^{-1} \pm 2.44\!\times\!10^{-1}$ \\
Whitley & \cellcolor{cBest}\textbf{$2.46\!\times\!10^{1} \pm 1.17\!\times\!10^{1}$} & $4.91\!\times\!10^{1} \pm 8.45$ & \cellcolor{cBest}\textbf{$4.41\!\times\!10^{2} \pm 8.39\!\times\!10^{1}$} & $6.27\!\times\!10^{2} \pm 6.23\!\times\!10^{1}$ & \cellcolor{cBest}\textbf{$1.46\!\times\!10^{3} \pm 1.93\!\times\!10^{2}$} & $2.07\!\times\!10^{3} \pm 1.14\!\times\!10^{2}$ \\
Exponential & \cellcolor{cBest}\textbf{$0 \pm 0$} & $4.62\!\times\!10^{-3} \pm 1.49\!\times\!10^{-3}$ & \cellcolor{cBest}\textbf{$1.05\!\times\!10^{-7} \pm 6.29\!\times\!10^{-8}$} & $4.26\!\times\!10^{-2} \pm 6.68\!\times\!10^{-3}$ & \cellcolor{cBest}\textbf{$1.48\!\times\!10^{-6} \pm 2.29\!\times\!10^{-6}$} & $9.46\!\times\!10^{-2} \pm 1.10\!\times\!10^{-2}$ \\
CosineMixture & \cellcolor{cBest}\textbf{$2.96\!\times\!10^{-2} \pm 5.91\!\times\!10^{-2}$} & $1.10\!\times\!10^{-1} \pm 3.15\!\times\!10^{-2}$ & \cellcolor{cBest}\textbf{$6.21\!\times\!10^{-1} \pm 2.57\!\times\!10^{-1}$} & $7.88\!\times\!10^{-1} \pm 1.05\!\times\!10^{-1}$ & \cellcolor{cBest}\textbf{$1.59 \pm 4.86\!\times\!10^{-1}$} & $1.65 \pm 1.97\!\times\!10^{-1}$ \\
\bottomrule
\end{tabular}}
\end{table}

%% file: outputs/tables/timing.tex
\begin{table}[t]
\centering
\caption{Wall-clock time in seconds (mean over 30 runs).}
\label{tab:timing}
\footnotesize
\setlength{\tabcolsep}{3.5pt}
\begin{tabular}{l cc cc cc}
\toprule
\rowcolor{cGroupHd}
Function & \multicolumn{2}{c}{$D=10$} & \multicolumn{2}{c}{$D=30$} & \multicolumn{2}{c}{$D=50$} \\
\cmidrule(lr){2-3}\cmidrule(lr){4-5}\cmidrule(lr){6-7}
\rowcolor{cGroupHd}
 & PSO & DPSO & PSO & DPSO & PSO & DPSO \\
\midrule
Ackley & $0.38$ & $0.44$ & $0.21$ & $0.26$ & $0.25$ & $0.30$ \\
Alpine1 & $0.40$ & $0.43$ & $0.22$ & $0.28$ & $0.26$ & $0.30$ \\
Bohachevsky & $0.36$ & $0.40$ & $0.19$ & $0.23$ & $0.28$ & $0.30$ \\
ChungReynolds & $0.33$ & $0.36$ & $0.17$ & $0.21$ & $0.20$ & $0.24$ \\
Cigar & $0.33$ & $0.37$ & $0.17$ & $0.21$ & $0.21$ & $0.25$ \\
CosineMixture & $0.38$ & $0.43$ & $0.21$ & $0.26$ & $0.25$ & $0.32$ \\
DixonPrice & $0.36$ & $0.40$ & $0.20$ & $0.25$ & $0.23$ & $0.27$ \\
Exponential & $0.38$ & $0.42$ & $0.21$ & $0.25$ & $0.24$ & $0.28$ \\
Griewank & $0.37$ & $0.41$ & $0.21$ & $0.25$ & $0.25$ & $0.29$ \\
HGBat & $0.39$ & $0.43$ & $0.21$ & $0.26$ & $0.27$ & $0.32$ \\
HappyCat & $0.38$ & $0.42$ & $0.21$ & $0.26$ & $0.24$ & $0.28$ \\
Levy & $0.38$ & $0.41$ & $0.22$ & $0.25$ & $0.26$ & $0.30$ \\
Pathological & $0.38$ & $0.41$ & $0.22$ & $0.26$ & $0.25$ & $0.29$ \\
Pinter & $0.41$ & $0.44$ & $0.26$ & $0.29$ & $0.30$ & $0.38$ \\
Qing & $0.40$ & $0.44$ & $0.21$ & $0.25$ & $0.25$ & $0.31$ \\
Quartic & $0.36$ & $0.40$ & $0.19$ & $0.23$ & $0.21$ & $0.26$ \\
Rastrigin & $0.39$ & $0.42$ & $0.21$ & $0.25$ & $0.24$ & $0.28$ \\
Rosenbrock & $0.35$ & $0.38$ & $0.19$ & $0.22$ & $0.22$ & $0.25$ \\
RotHyperEllipsoid & $0.34$ & $0.39$ & $0.17$ & $0.21$ & $0.20$ & $0.23$ \\
Salomon & $0.37$ & $0.40$ & $0.19$ & $0.24$ & $0.23$ & $0.27$ \\
SchafferF6 & $0.39$ & $0.43$ & $0.23$ & $0.28$ & $0.28$ & $0.29$ \\
Schwefel1.2 & $0.36$ & $0.39$ & $0.19$ & $0.23$ & $0.23$ & $0.27$ \\
Schwefel2.20 & $0.35$ & $0.39$ & $0.18$ & $0.22$ & $0.21$ & $0.25$ \\
Schwefel2.21 & $0.36$ & $0.39$ & $0.18$ & $0.22$ & $0.21$ & $0.25$ \\
Schwefel2.22 & $0.35$ & $0.39$ & $0.18$ & $0.23$ & $0.22$ & $0.26$ \\
Schwefel2.23 & $0.35$ & $0.40$ & $0.19$ & $0.23$ & $0.21$ & $0.25$ \\
Schwefel & $0.36$ & $0.39$ & $0.22$ & $0.23$ & $0.23$ & $0.27$ \\
Sphere & $0.38$ & $0.41$ & $0.18$ & $0.21$ & $0.21$ & $0.25$ \\
StretchedV & $0.38$ & $0.43$ & $0.23$ & $0.26$ & $0.25$ & $0.31$ \\
SumDiffPowers & $0.33$ & $0.36$ & $0.18$ & $0.22$ & $0.22$ & $0.27$ \\
SumSquares & $0.35$ & $0.38$ & $0.18$ & $0.22$ & $0.21$ & $0.25$ \\
Wavy & $0.39$ & $0.43$ & $0.22$ & $0.27$ & $0.28$ & $0.32$ \\
Weierstrass & $0.44$ & $0.49$ & $0.27$ & $0.32$ & $0.47$ & $0.57$ \\
Whitley & $0.44$ & $0.54$ & $0.28$ & $0.34$ & $0.43$ & $0.54$ \\
XinSheYang2 & $0.37$ & $0.42$ & $0.21$ & $0.26$ & $0.29$ & $0.30$ \\
Zakharov & $0.37$ & $0.40$ & $0.20$ & $0.24$ & $0.23$ & $0.27$ \\
\bottomrule
\end{tabular}
\end{table}

%% file: references.bib
@inproceedings{kennedy1995particle,
  title={Particle swarm optimization},
  author={Kennedy, James and Eberhart, Russell},
  booktitle={Proceedings of ICNN'95-international conference on neural networks},
  volume={4},
  pages={1942--1948},
  year={1995},
  organization={ieee}
}

@article{wang2018particle,
  title={Particle swarm optimization algorithm: an overview},
  author={Wang, Dongshu and Tan, Dapei and Liu, Lei},
  journal={Soft computing},
  volume={22},
  number={2},
  pages={387--408},
  year={2018},
  publisher={Springer}
}

@inproceedings{shi1998modified,
  title={A modified particle swarm optimizer},
  author={Shi, Yuhui and Eberhart, Russell},
  booktitle={1998 IEEE international conference on evolutionary computation proceedings. IEEE world congress on computational intelligence (Cat. No. 98TH8360)},
  pages={69--73},
  year={1998},
  organization={IEEE}
}

@book{van2001analysis,
  title={An analysis of particle swarm optimizers},
  author={Van Den Bergh, Frans},
  year={2001},
  publisher={University of Pretoria (South Africa)}
}

@article{fogel1994evolutionary,
  title={An introduction to simulated evolutionary optimization},
  author={Fogel, David B},
  journal={IEEE transactions on neural networks},
  volume={5},
  number={1},
  pages={3--14},
  year={1994},
  publisher={IEEE}
}

@article{tang2013surrogate,
  title={A surrogate-based particle swarm optimization algorithm for solving optimization problems with expensive black box functions},
  author={Tang, Yuanfu and Chen, Jianqiao and Wei, Junhong},
  journal={Engineering optimization},
  volume={45},
  number={5},
  pages={557--576},
  year={2013},
  publisher={Taylor \& Francis}
}

@inproceedings{helwig2007dimension,
  title={Particle swarm optimization in high-dimensional bounded search spaces},
  author={Helwig, Sabine and Wanka, Rolf},
  booktitle={2007 IEEE swarm intelligence symposium},
  pages={198--205},
  year={2007},
  organization={IEEE}
}

@inproceedings{hu2002solving,
  title={Solving constrained nonlinear optimization problems with particle swarm optimization},
  author={Hu, Xiaohui and Eberhart, Russell and others},
  booktitle={Proceedings of the sixth world multiconference on systemics, cybernetics and informatics},
  volume={5},
  pages={203--206},
  year={2002},
  organization={Citeseer}
}

@article{gang2012novel,
  title={A novel particle swarm optimization algorithm based on particle migration},
  author={Gang, Ma and Wei, Zhou and Xiaolin, Chang},
  journal={Applied Mathematics and Computation},
  volume={218},
  number={11},
  pages={6620--6626},
  year={2012},
  publisher={Elsevier}
}

@article{yildiz2009novel,
  title={A novel particle swarm optimization approach for product design and manufacturing},
  author={Y{\i}ld{\i}z, Ali R{\i}za},
  journal={The International Journal of Advanced Manufacturing Technology},
  volume={40},
  pages={617--628},
  year={2009},
  publisher={Springer}
}

@article{hakli2014novel,
  title={A novel particle swarm optimization algorithm with Levy flight},
  author={Hakl{\i}, H{\"u}seyin and U{\u{g}}uz, Harun},
  journal={Applied Soft Computing},
  volume={23},
  pages={333--345},
  year={2014},
  publisher={Elsevier}
}

@article{eschwege2024belief,
  title={Belief space-guided approach to self-adaptive particle swarm optimization},
  author={Eschwege, Daniel von and Engelbrecht, Andries},
  journal={Swarm Intelligence},
  volume={18},
  number={1},
  pages={31--78},
  year={2024},
  publisher={Springer}
}

@article{clerc2002particle,
  title={The particle swarm--explosion, stability, and convergence in a multidimensional complex space},
  author={Clerc, Maurice and Kennedy, James},
  journal={IEEE Transactions on Evolutionary Computation},
  volume={6},
  number={1},
  pages={58--73},
  year={2002},
  publisher={IEEE}
}

@techreport{duchi2007derivations,
  title={Derivations for linear algebra and optimization},
  author={Duchi, John},
  year={2007},
  institution={University of California, Berkeley}
}

@article{jamil2013literature,
  title={A literature survey of benchmark functions for global optimisation problems},
  author={Jamil, Momin and Yang, Xin-She},
  journal={International Journal of Mathematical Modelling and Numerical Optimisation},
  volume={4},
  number={2},
  pages={150--194},
  year={2013},
  publisher={Inderscience Publishers}
}
